\documentclass[11pt,a4paper]{article}

\usepackage{amsmath,amstext,amsthm,amssymb,amsfonts}
\usepackage{graphicx,psfrag}
\usepackage{comment} 
\usepackage{tikz}
\usepackage{tkz-graph}
\tikzstyle{vertex}=[circle, draw, inner sep=0pt, minimum size=6pt]

\usepackage{bm}
\usepackage[margin=1in]{geometry}
\usepackage{floatrow}
\usepackage{color,xspace}
\usepackage[shortlabels]{enumitem}
\usetikzlibrary{positioning}
\usepackage{authblk}

\newfloatcommand{capbtabbox}{table}[][\FBwidth]
\usepackage{algorithm}% http://ctan.org/pkg/algorithm
\usepackage{algpseudocode}% http://ctan.org/pkg/algorithmicx

\newtheorem{theorem}{Theorem}

\newtheorem{lemma}{Lemma}
\newtheorem{definition}{Definition}
\newtheorem{corollary}{Corollary}

\newtheorem{claim}{Claim}
\newtheorem*{theorem*}{Theorem}
\newtheorem*{lemma*}{Lemma}

\newcommand{\hf}{\hat{f}}
\newcommand{\tf}{\tilde{f}}
\newcommand{\hg}{\hat{g}}
\newcommand{\tg}{\tilde{g}}
\newcommand{\nphard}{NP-hard}

\newcommand{\mcD}{\mathcal{D}}
\newcommand{\mcU}{\mathcal{U}}
\newcommand{\mcC}{\mathcal{C}}
\newcommand{\mcF}{\mathcal{F}}

\newcommand{\ip}{\textsf{IP}}
\newcommand{\op}{\textsf{OP}}
\newcommand{\im}{\textsf{IM}}
\newcommand{\all}{\textsf{G}}

\DeclareMathOperator{\fix}{\textsf{fix}}

\DeclareMathOperator{\val}{\textsf{val}}

\DeclareMathOperator{\conv}{Conv}

\DeclareMathOperator{\poly}{poly}
\newtheorem{result}{Result}
\title{Partial Function Extension with Applications to Learning and Property Testing}
\author{Umang Bhaskar\thanks{Work supported in part by a Ramanujan fellowship. Email: \texttt{umang@tifr.res.in}} }
\author{Gunjan Kumar\thanks{Email: \texttt{gunjan.kumar@tifr.res.in}}}
\affil{Tata Institute of Fundamental Research, Mumbai}
\begin{document}
	\maketitle

	\begin{abstract}
		In \emph{partial function extension}, we are given a partial function consisting of $n$ points from a domain and a function value at each point. Our objective is to determine if this partial function can be extended to a function defined on the domain, that additionally satisfies a given property, such as convexity. This basic problem underlies research questions in many areas, such as learning, property testing, and game theory. We formally study the problem of extending partial functions to satisfy fundamental properties in combinatorial optimization, focusing on upper and lower bounds for extension and applications to learning and property testing.
		
		\begin{itemize}
			\item For subadditive functions, we show the extension problem is coNP-complete, and we give tight bounds on the approximability. We also give an improved lower bound of $\Omega(\sqrt{m})$ for learning subadditive functions. Previously, Balcan et al. (2012) gave a lower bound of $\Omega(\sqrt{m}/\log m)$ for this problem. We also give the first nontrivial testers for subadditive and XOS functions. 
			
			\item For submodular functions, we show that if a partial function can be extended to a submodular function on the lattice closure\footnote{The \emph{lattice closure} $LC(\mcD)$ of set of points $\mcD$ is the minimal set that contains $\mcD$ and is closed under union and intersection.} of the partial function, it can be extended to a submodular function on the entire domain. We obtain algorithms for determining extendibility in a number of cases, including if $n$ is a constant, or the points are nearly the same size. The result uses a combinatorial certificate for non-extendibility which we call a square certificate. Seshadhri and Vondrak (2014) previously  give a characterization in terms of path certificates. The complexity of extendibility is in general unresolved.
			
			\item Lastly, for convex functions in $\mathbb{R}^m$, we show an interesting juxtaposition: while we can determine the existence of an extension efficiently, computing the value of a widely-studied convex extension at a given point  is strongly NP-hard.
		\end{itemize}
	\end{abstract}

	\thispagestyle{empty}

	\clearpage

	\pagenumbering{arabic}
	\section{Introduction}

A \emph{partial function} consists of a set $\mcD$ of points from a domain, and a real value at each of the points. Given a property $P$, the \emph{partial function extension} problem is to determine if there exists a \emph{total} function $f$  ($f$ is defined on the entire domain) that \emph{extends} the partial function ($f$ equals the given value at each point in $\mcD$) and satisfies $P$. E.g., property $P$ could be linearity, and we are required to determine if there exists a linear function that extends the given partial function. In this paper, we study partial function extension when $\mcD$ is finite, to fundamental properties in combinatorial optimization.

The problem of partial function extension underlies research and techniques in a number of different areas, and is hence intensely studied. We mention three such areas.  Firstly, in \emph{property testing}, a function is given by an oracle, and the problem is to determine with high probability by querying the oracle whether the function satisfies a required property, or is far from it. The focus in property testing is on algorithms with optimal query-complexity. 
Typically a testing algorithm cleverly queries some sets and rejects  if the values at the queried sets cannot be extended to a function with the required property. Clearly characterizing when a partial function can be extended  plays an important role here. The problem of partial function extension, and its connection to property testing, is also explicitly raised by Seshadhri and Vondrak~\cite{submodularity}. 

Secondly, in \emph{learning theory}, the goal is to understand if a family of functions can be learned by random samples. That is, does there exist an efficient algorithm that for any target function in the family takes as input the function values at a set of sampled points, and returns a function that is ``close'' to the target function? Here, partial function extension can be used to give lower bounds on the learnability of various function classes (and has been used thus in previous papers, e.g., Balcan, Constantin, Iwata and Wang~\cite{BalcanCIW12}). 

Thirdly, in economics, given data from experiments regarding agent behaviour (such as the purchases made by an agent, or the bids of a (truthful) bidder in auctions), \emph{revealed preference theory} studies whether the data is rationalizable by utility functions with a particular property. That is, whether there exists a utility function with a particular property that is consistent with the observed data. A natural assumption for utility functions is ``diminishing marginal returns'', which translates to submodularity for indivisible goods, and concavity for divisible goods~\cite{LehmannLN06}. Other assumptions on utility functions are also common, e.g., subadditivity, XOS, etc~\cite{BalcanCIW12,BhawalkarR11}. The problem of deciding rationalizability by utility functions with these properties is exactly the problem we study.

In each of the above areas, a basic step towards a solution is often to determine if a given partial function can be extended to a total function. In this paper, rather than a means to an end, we study the problem of partial function extension itself. We focus on the complexity of deciding if a partial function can be extended to functions satisfying fundamental properties --- subadditivity, XOS, submodularity, and convexity. These represent perhaps the most commonly studied classes of functions in all of combinatorial optimization. In obtaining our results on partial function extension, we show that the structural lemmas can be used to obtain several results for property testing and learning, thus validating a direct study of partial function extension.

Formally, a \emph{partial function} is a set of duples	$H = \{ (T_1,f_1),$ $(T_2,f_2),$ $\dots,$ $(T_n,f_n)\}$, with $T_i$ in the domain $\{0,1\}^m$ or $\mathbb{R}^m$, and $f_i \in \mathbb{R}$ the observed function value at $T_i$. Additionally, we are given a property $P$. The \emph{$P$-Extension} problem is to determine if there exists a total function $f$ defined on the domain $\{0,1\}^m$ or $\mathbb{R}^m$ that satisfies property $P$ and extends the given partial function $H$, i.e., $f(T_i) = f_i$ for all $i \in \{1, \dots, n\}$.  We also consider the \emph{Approximate $P$-Extension} problem, where we want to determine the minimum multiplicative error for a given partial function to extend to a function that satisfies the given property. That is, in Approximate $P$-Extension, we want to find the minimum $\alpha \ge 1$ such that  a function $f$ satisfies property $P$ and additionally, $f_i \le f(T_i) \le \alpha f_i$ for all $i \in \{1, \dots, n\}$. 
If Approximate $P$-Extension is computationally hard, we are interested in approximating $\alpha$.

Note that in our case, our input is $H$. An algorithm is efficient if it runs in time polynomial in the size of $H$, which may be exponential in the dimension $m$. 

We note a basic difference between partial function extension on the one hand, and property testing and learning on the other. In partial function extension, there is no target function $f^*$; we are interested in determining if \emph{any} total function that extends the given partial function has the required property. In property testing and learning, there exists a target function $f^*$ which we access via an oracle (in property testing) or via samples from a distribution (in learning). 

A function $f:2^{[m]} \rightarrow \mathbb{R}_{\ge 0}$ is subadditive  if $f(A) + f(B) \ge f(A \cup B)$ for all sets $A$ and $B$.   A function $f:2^{[m]} \rightarrow \mathbb{R}_{\ge 0}$ is an XOS function if it can be expressed as the maximum of $k$ linear functions for some $k \ge 1$. XOS functions are a subclass of subadditive functions. A function $f$ is submodular if $f(A) + f(B) \ge f(A \cup B) + f(A \cap B)$ for all $A,B \subseteq [m]$. 

For notation, $\mcD := \{T_i\}_{i \in [n]}$ is the set of points in the given partial function $H$. These are called \emph{defined} points, and $\mcU := 2^{[m]} \setminus \mcD$ are \emph{undefined} points. Points on the hypercube $\{0,1\}^m$ are naturally subsets of $[m]$, and for $S \subseteq [m]$, $\chi(S) \in \{0,1\}^m$ is its characteristic vector. We frequently use this correspondence. All missing proofs are in the appendices.

\paragraph*{Our Contribution.} 

We show the following main results.

\begin{result}
Subadditive Extension is coNP-complete. There is an $O(\log m)$ approximation algorithm for Approximate Subadditive Extension, and if $P \neq NP$, this is tight.
\end{result}

The lower bounds in the theorem depend upon characterizations of partial functions that can be extended to subadditive functions. The upper bound uses the fact that for XOS functions, a well-studied subclass of subadditive functions, Approximate  Extension (and hence  Extension) can be solved in polynomial time. Further, any subadditive function can be approximated by an XOS function, by a factor of $O(\log m)$~\cite{BhawalkarR11,Dobzinski07}.

 Our characterization for subadditive functions, as well as known characterizations for XOS functions, can be used to give the following results for learning and property testing.

\begin{result}
Subadditive functions cannot be learned by a factor of $o(\sqrt{m})$.
\end{result}

This improves upon a previous lower bound of $\Omega(\sqrt{m}/\log m)$~\cite{BalcanCIW12}. We combine the characterization of subadditive functions with recent results on the size of combinatorial families of sets called $r$-cover free families~\cite{erdos1985families,d2014bounds} for the above result.

\begin{result}
Given $\epsilon > 0$, there are testers for subadditive and XOS functions that make $2^{m/2 + O(\sqrt{m \log (1/\epsilon)})}$ queries. Further, there is a tester for nonmonotone subadditive functions that makes $2^{O(\sqrt{m\log (1/\epsilon)}\log m)}$ queries.
\end{result}

We thus obtain the first nontrivial testers for subadditive and XOS functions. 

For submodular functions, we show the following main result. Given $\mcF \subseteq \{0,1\}^m$, we say a function $f$ is submodular in $\mcF$ if $f(A) + f(B) \ge f(A \cup B) + f(A \cap B)$ for all $A, B, A \cup B, A \cap B \in \mcF$.  

\begin{result}
For partial function $H$ with defined points $\mcD$, let $\mcF$ be the family of sets that are (i) both contained in and contained by some set in $\mcD$, and (ii) obtained by the union and intersection of sets in $\mcD$ (are in the lattice closure of $\mcD$). Then the partial function is extendible to a submodular function in $\{0,1\}^m$ iff it can be extended to a submodular function in $\mcF$.
\end{result}

Thus if $|\mcF|$ is $poly(m,|\mcD|)$ then Submodular Extension can be solved in polynomial time. This includes the case when $|\mcD|$ is a constant, when all points in $\mcD$ have size difference $O(\log m)$, and when $\mcD$ is an antichain. Further, if $\mcD$ is an antichain, then $\mcF = \mcD$ and 
  any assignment of values to the points in $\mcD$ is trivially submodular in $\mcD$. Hence if $\mcD$ is an antichain, it can always be extended to a submodular function. 
   Our results for submodular functions depend on a combinatorial characterization of nonextendibility, which we call a \emph{square certificate}. Our proof of the above result and development of square certificates forms our main technical contribution. Seshadhri and Vondrak~\cite{submodularity} also study property testing of submodular functions, and develop an alternative certificate called a \emph{path certificate}. We believe that square certificates are more natural, and may lead to improved testers for submodular functions. In general the problem of Submodular Extension remains open. However, we can use our result for the extendibility of antichains to obtain the following result for learning submodular functions %in the PMAC model.

\begin{result}
Submodular functions cannot be learned.\footnote{In contrast, the class of nonnegative, monotone submodular functions can be learned with approximation ratio $O(\sqrt{m})$~\cite{BalcanH11}.}
\end{result}

Lastly, we consider convex functions. The problem of Convex Extension has been studied before in convex analysis (e.g.,~\cite{DragomirescuI92,Yan12}). We however show an interesting juxtaposition of results: While it can be determined in polynomial time if a partial function is extendible to a convex function, determining the value of a natural and widely-studied extension at a given point is NP-hard.

\begin{result}
Approximate Convex Extension (and hence Convex Extension) is in $P$. However, determining the value of a canonical  extension at a given point is strongly NP-hard.
\end{result}

	\section{Subadditive and XOS Functions}
We now consider the problem of extending a given partial function $H$ to monotone subadditive and XOS functions. A function $f:2^{[m]} \rightarrow \mathbb{R}_{\ge 0}$ is subadditive if $f(A) + f(B) \ge f(A \cup B)$ for all sets $A$ and $B$, and monotone if  $f(A) \ge f(B)$ for all $A \supseteq B$. A function $f:2^{[m]} \rightarrow \mathbb{R}_{\ge 0}$ is an XOS function if it can be expressed as the maximum of $k$ linear functions for some $k \ge 1$, i.e., there exist vectors $w_{i} \in \mathbb{R}_{\ge 0}^m$ for $1 \le i \le k$  such that $f(S) = \max w_{i}^T  \chi(S)$ for every $S \subseteq [m]$. XOS functions are a subclass of subadditive functions and  are equivalent to fractionally subadditive functions \cite{feige2009maximizing}. A function $f:2^{[m]} \rightarrow \mathbb{R}_{\ge 0}$ is fractionally subadditive if $f(T) \le \sum_{S} \lambda_S f(S)$ for all $T$ such that $\lambda_S \ge 0$ and $\sum_{S:s \in S} \lambda_S \ge 1$ for each $s \in T$. 
Subadditive functions capture the important case of complement-free functions, for which no two subsets of the ground set $[m]$ ``complement'' each other. This is a natural assumption in many applications, and hence these functions and various subclasses, including XOS functions, are widely used in game theory~\cite{BalcanCIW12,BhawalkarR11,LehmannLN06}. 

For XOS functions, the following positive result follows from writing a linear program with the vectors $(w_i)_{i \le k}$ with $w_i \in \mathbb{R}^m_+$ as variables. While in general $k$ may be exponential, 
a partial function $H$ is extendible iff  the linear program is feasible for $k = n$.

 \begin{theorem}
 	\label{gen-ext-xos}
 Aproximate XOS Extension (and hence XOS Extension) can be solved efficiently.
 \end{theorem}

\begin{comment}
The input is a set $H = \{(T_1,f_1),\dots,(T_n,f_n)\}$   with $T_i \subseteq [m]$ and $f_i \ge 0$.  
And the goal of extension and generalised-extension problem are as before. 
 Let $\mathcal{D}$ be the set $\{T_1,\dots,T_n\}$.
\end{comment}

 We give the following characterization for subadditive functions, also  implicit in  Lemma 3.3 of \cite{badanidiyuru2012sketching}.
\begin{lemma}[\cite{badanidiyuru2012sketching}]
	\label{lemma-subadditive}
	Partial function $H$ is extendible to a subadditive function iff  $\sum_{i = 1}^{r} f(T_i) \ge f(T_{r+1})$ for all $T_1,\dots,T_r,T_{r+1} \in \mathcal{D}$ such that $\cup_{i = 1}^{r} T_i \supseteq T_{r+1}$.
\end{lemma}
\begin{proof}	
	For the first direction, if there exist $T_1,\dots,T_r,T_{r+1} \in \mathcal{D}$ such that  $\cup_{i = 1}^{r} T_i \supseteq T_{r+1}$ and $\sum_{i = 1}^{r} f(T_i) < f(T_{r+1})$ then either $f(T_{r+1}) > f(\cup_{i = 1}^{r} T_i)$ or $f(\cup_{i = 1}^{r} T_i) > \sum_{i = 1}^{r} f(T_i)$, and hence either monotonicity or subadditivity  is violated. For the other direction, first assume that $\cup_{i = 1}^{n} T_i = [m]$. Then the function $\hat{f}(S) =  \min \{\sum_{i = 1}^{r} f(T_i) | S \subseteq \cup_{i = 1}^{r} T_i, \thinspace T_i \in \mathcal{D} \thinspace \forall i \in [r]\}$ can be seen to be  monotone  subadditive extension. If $\cup_{i = 1}^{n} T_i \subsetneq [m]$ then the function  $\tilde{f}$ is a monotone subadditive extension where  $\tilde{f}(S) = \hat{f}(S)$ for all $S  \subseteq \cup_{i = 1}^{n} T_i$, and otherwise  $\tilde{f}(S) = \hat{f}(S')$ where $S' = S \cap\cup_{i = 1}^{n}T_i$.
\end{proof}
We immediately obtain the following result.

\begin{corollary}
	 Subadditive Extension  is in \textsf{coNP}, and be solved in $poly(m,2^n)$ time.
\end{corollary}
We use the characterization  to show that $\Theta(\log m)$ is a tight bound on the approximability of Approximate Subadditive Extension, unless $P = NP$ (and that the Extension problem is \textsf{coNP}-complete). For the lower bound, we give a reduction from \textsf{Set-Cover}.  For the upper bound, we use earlier results which show that any subadditive function can be $O(\log m)$-approximated by an XOS function~\cite{BhawalkarR11,Dobzinski07}. Since Approximate XOS Extension can be efficiently solved (Theorem \ref{gen-ext-xos}), this gives us our upper bound.

 \begin{theorem}
 		\label{gen-ext-subadditive-approx}
Subadditive Extension is \textsf{coNP}-complete. There is an $O(\log m)$ approximation algorithm for Approximate Subadditive Extension, and if $P \neq NP$, this is optimal.
 \end{theorem}
 \begin{proof}
 	Recall the \textsf{Set-Cover} problem. An instance of \textsf{Set-Cover} is  a universe $[m]$, family of sets $V = \{S_1,\dots,S_n\}$ such that $S_i \subseteq [m]$ and an integer $k$. We need to determine if there exists a cover of universe $[m]$ of size at most $k$.
 	
 	First we prove that the Subadditive Extension is CoNP-hard by reduction from \textsf{Set-Cover}.  Construct a partial function that is defined on each set in  $V$ and $[m]$. The value at each set  $S_i \in V$ is $1$, and the value at $[m]$ is $k+1$.   If this partial function can be extended then every cover of $[m]$ must have size at least $k+1$. On the other hand, if the partial function can not be extended then there must exist a cover of size at most $k$. Both of the above facts easily follow from Lemma \ref{lemma-subadditive}.
 	
 	 The lower bound of $\Omega(\log m)$ for Approximate  Extension, as before, the partial function is defined on  sets $V \cup [m]$,  and value at each set in $V$ is $1$, and at $[m]$ is $m$. Suppose we have a $\rho$ approximation algorithm for Approximate Extension, which for this instance returns value $\beta$. Then note that $\alpha^* \ge \beta/\rho$, where $\alpha^*$ is the optimal value of $\alpha$ for the Approximate  Extension.
 	
 	Since the algorithm returns value $\beta$, so there exists an extension $f$ such that $1 \le f(S_i) \le \beta$ for all $S_i \in V$ and $f([m]) \ge m$. Therefore, by Lemma \ref{lemma-subadditive},    every cover of $[m]$ has size at least $m/\beta$. 
 	Now we claim that there must exist a cover with size at most $ m \rho/\beta$. If not, then  all covers of $[m]$ have size at least $\gamma > m \rho/\beta$,  and it is easy to see that the partial function $\{(S_1,m/\gamma),\dots,(S_n,m/\gamma),([m],m)\}$ is extendible by Lemma \ref{lemma-subadditive}. This implies $m/\gamma \ge \alpha^* \ge \beta/\rho$ which is a contradiction.
 	This then gives an $\rho$-approximation algorithm for Set Cover, and since Set Cover cannot be approximated by a factor better than $(1 - \epsilon) \log m$~\cite{DinurS14}, this is true of Approximate Extension also. 
 	
 	The upper bound for Approximate Extension is shown in Appendix \ref{sec:subadditiveappendix}
 \end{proof}
 \paragraph*{A lower bound on learning subadditive functions.} We now show that subadditive functions cannot be learned by a factor of $o(\sqrt{m})$ in the PMAC model of learning. The PMAC (Probably Mostly Approximate Correct) model seeks to determine for a family $\mathcal{F}$ of functions, if it is possible to efficiently obtain a function $f$ ``close to'' a target function $f^* \in \mathcal{F}$, given samples from some distribution over $2^{[m]}$ and the value of $f^*$ at the sampled points. Formally, let $\mathcal{F} \subseteq \{f | f :2^{[m]} \rightarrow \mathbb{R}\}$ be a family of set functions (e.g., subadditive functions).
 \begin{definition}[\cite{balcan2011learning}]
 	An Algorithm $\mathcal{A}$ PMAC-learns a family of functions $\mathcal{F} $ with approximation factor $\alpha$, if for \emph{any} distribution $\mu$ (on $2^{[m]}$) and \emph{any} target function $f^* \in \mathcal{F}$,  and  for \emph{any} sufficiently small $\epsilon,\delta >0$:
 	\begin{itemize}
 		\item $\mathcal{A}$ takes the sequence  $\{(S_i,f^*(S_i))\}_{1 \le i \le l}$  as input where  $l$ is $poly(m,1/\delta,1/\epsilon)$ and the sequence $\{S_i\}_{1 \le i \le l}$ is drawn i.i.d. from the distribution $\mu$, % The value of $l$ must be $poly(m,1/\delta,1/\epsilon)$ %number of samples drawn from distribution $D^*$ (called training examples) as input.
 		\item $\mathcal{A}$  runs in $poly(m,1/\delta,1/\epsilon)$ time,
 		\item $\mathcal{A}$  returns a function $f:2^{[m]} \rightarrow \mathbb{R}$ such that 
 		\[
 		Pr_{S_1,\dots,S_l \sim D^*} \big[Pr_{S \sim D^*} [f^*(S) \le f(S) \le \alpha f^*(S)] \ge 1 - \epsilon \big] \ge 1 - \delta.
 		\]
 	\end{itemize}
 \end{definition}
 That is, with at least $1 - \delta$ probability (over examples drawn from $\mu$), the value of the returned function $f$  should be within an $\alpha$ factor of the target function $f^*$ for at least $1 - \epsilon$ fraction of the probability mass according to $\mu$. 
 % Note that the distribution $D^*$ is unknown but fixed. 
 The following lemma makes explicit the connection between PMAC-learning and extending partial functions (we use it later in showing lower bounds for learning submodular functions as well). The lemma has been implicitly used earlier to obtain lower bounds  on learning  subadditive and submodular functions~\cite{balcan2012learning,balcan2011learning}.
 \begin{lemma}
 	\label{learning-application}
 	Suppose there exists a family $\mathcal{D} = \{T_1,\dots,T_n\}$  of subsets of $[m]$ such that $n$ is superpolynomial in $m$, and the partial function  $H = \{(T_1,f_1),\dots,(T_n,f_n)\}$ is extendible to a function in $\mathcal{F}$ for any values of $f_i \in [1,r]$ (where $r \ge 1$), $i \in [n]$. Then the family of functions $\mathcal{F}$ cannot be learned by any factor $< r $.
 \end{lemma} 
  The above bound holds even if the algorithm knows the distribution $\mu$, is allowed unbounded computation and chooses samples adaptively. 
  
   Balcan et. al.~\cite{balcan2012learning} proved an upper bound of $O(\sqrt{m}\log m)$ and an $\Omega(\sqrt{m}/\log m)$ lower bound for learning subadditive functions. Using Lemmas~\ref{lemma-subadditive},~\ref{learning-application} and a known result for cover-free families~\cite{erdos1985families,d2014bounds}, we show an improved lower bound of $\Omega(\sqrt{m})$.
  
  A family of sets $\mathcal{F} \subseteq 2^{[m]}$ is called an \emph{$r$-cover free} family \cite{erdos1985families,d2014bounds} if for all distinct sets $A_1,\dots,A_r,A_{r+1} \in \mathcal{F}$ we have $A_{r+1} \not \subseteq \cup_{i =1}^{r} A_i$. Let $f_r(m)$ be the  cardinality of the largest $r$-cover free family.
  
  \begin{theorem}[\cite{d2014bounds}]
  	\label{Erdos}
  	$f_r(m) = 2^{\Theta\left(\frac{m \log r}{r^2}\right)}$.
  \end{theorem}
  \begin{lemma}
  	\label{coverfree-partialfunction}
  	If $\mathcal{D} = \{T_1,\dots,T_n\}$ is an $r$-cover free family then the partial function $\{(T_1,f_1),\dots,(T_n,f_n)\}$ is extendible to a subadditive function for any value of $f_i \in [1,r+1], i \in [n]$.
  \end{lemma}
  \begin{proof}
  	Suppose the partial function is not extendible.  Therefore, by Lemma \ref{lemma-subadditive}, there exists sets $T_1,\dots,T_k,T_{k+1}$ for some $k \ge 1$ such that $T_{k+1} \subseteq \cup_{i = 1}^{k} T_i$ and $f_{k+1} > \sum_{i = 1}^{k} f_i$.  Therefore, we have  $r+1 \ge f_{k+1} > k  $ which is a contradiction as  $\mathcal{D}$ is an $r$-cover free family.
  \end{proof}
  \begin{theorem}
  	\label{lb-subadditive}
  	In the PMAC model, subadditive functions cannot be learned by any $o(\sqrt{m})$ factor. 
  \end{theorem}
  \begin{proof}
  	We have $f_r(m) \ge 2^{\frac{c m \log r}{r^2}} = m^{\frac{cm}{r^2}(\frac{1}{2} - \frac{\log(\sqrt{m}/r)}{\log m})}$ for some constant $c$. For $r \le m^{1/4}$, $f_r(m)$ is clearly superpolynomial in $m$. Also, for $r \ge m^{1/4}$, $\frac{1}{2} - \frac{\log(\sqrt{m}/r)}{\log m} \ge \frac{1}{4}$. Hence, $2^{\frac{cm \log r}{r^2}}$ is superpolynomial for $r = o(\sqrt{m})$. Therefore, for any such $r$, by Theorem \ref{Erdos} there exists a $r$-cover free family $\mathcal{D} = \{T_1,\dots,T_n\}$ such that $n$ is superpolynomial. The theorem is directly implied by  Lemmas  \ref{learning-application} and \ref{coverfree-partialfunction}.	
  \end{proof}

 \paragraph*{Testers for  subadditive and XOS functions.} We now describe testers for  subadditive and XOS functions that make $2^{m/2 + O(\sqrt{m \log (1/\epsilon)})}$ queries; these are the first non-trivial testers for either of these functions. For definitions, we focus on subadditive functions, but the corresponding definitions for XOS functions are obvious. 
 A function $f :2^{[m]} \rightarrow \mathbb{R}_{\ge 0}$ is $\epsilon$-far from subadditive if for any subadditive function $g:2^{[m]} \rightarrow \mathbb{R}_{\ge 0}$,  we have $|S \subseteq [m]: f(S) \neq g(S)| \ge  \epsilon 2^m$. A tester for subadditive functions is a randomized algorithm that takes distance parameter $\epsilon$ and oracle access to a function  $f :2^{[m]} \rightarrow \mathbb{R}_{\ge 0}$ as inputs, and accepts if $f$ is subadditive, and rejects with constant probability if $f$ is $\epsilon$-far from subadditive.
 
 We describe the testers here. 
  Let $\lambda = \sqrt{\log (4/\epsilon)}$, and define $M_\lambda = \{S \subseteq [m] |  |S| \le m/2 + \lambda \sqrt{m}\}$. 
  % Note that the set $M_\lambda$ is downward closed, i.e., if $S \in M_\lambda$ and $S' \subseteq S$, then $S' \in M_\lambda$. 
   The tester repeats the following steps $1/\epsilon$ times.
  \begin{itemize}
  	\item Randomly pick a set $T \in M_\lambda$ and query  the sets  $Q(T) = \{S | S \subseteq T \} $.
  	\item (Subadditive tester) Reject if EITHER there exists $T' \in Q(T)$ such that $f(T)  < f(T')$ OR there exist $T_1,\dots,T_r \in  Q(T)$ for some $r \ge 1$ such that $T = \cup_{i =1}^{r} T_i$ and $f(T) > \sum_{i =1}^{r} f(T_i)$.
  	\item (XOS tester) Reject if EITHER there exists $T' \in Q(T)$ such that $f(T)  < f(T')$ OR there exist $T_1,\dots,T_r \in  Q(T)$ and $\alpha_1, \dots, \alpha_r \in \mathbb{R}_+$ for some $r \ge 1$ such that for all elements $s \in T$, $\sum_{j:s \in T_j} \alpha_j \ge 1$ and $f(T) > \sum_{j =1}^{r} \alpha_j f(T_j)$.
  \end{itemize}

 It is clear that tester makes $|Q(T)|/\epsilon$ queries and $|Q(T)| \le 2^{m/2 + \lambda \sqrt{m}} \le 2^{m/2 + O(\sqrt{m \log (1/\epsilon)})}$, which is also a bound on the number of queries by the tester.
 
   Recall that XOS functions are  fractionally subadditive functions, i.e.,  a function $f:2^{[m]} \rightarrow \mathbb{R}_{\ge 0}$ is XOS iff for all  $T$  and $\alpha_S \ge 0$ such that  $\sum_{S:s \in S} \alpha_S \ge 1$ for each $s \in T$,  $f(T) \le \sum_{S} \alpha_S f(S)$. Clearly if the function $f$ is  subadditive or XOS then the tester accepts.   Now we need to show if the function $f$ is $\epsilon$-far from  subadditive or XOS then tester rejects with high probability. Before that
   we show the following characterization for extension of a partial function to an XOS function. 
    \begin{lemma}
  	\label{lemma-XOS}
  	 Partial function $H$ is extendible to a XOS function  iff for all $T \in \mathcal{D}$ and  all $\alpha_S \ge 0$ for all $S \in \mathcal{D}$ such that $\sum_{S \in \mathcal{D}: s \in S} \alpha_S \ge 1$ for each $s \in T$,  we have  $f(T) \le \sum_{S \in \mathcal{D}} \alpha_S f(S)$.
  \end{lemma}
  \begin{proof}
  	Obviously if there is a XOS extension then the above property must hold as XOS functions are also fractionally subadditive.  For the other direction, consider $ \hat{f}(T) =  \min \{\sum_{S \in \mathcal{D}} \alpha_S f(S) |\\ \sum_{S \in \mathcal{D}} \alpha_S \chi(S) \ge \chi(T) , \thinspace \alpha_S \ge 0 \}$. We claim that  $ \hat{f}(T)$ is a XOS extension of the partial function. Let $T \in \mathcal{D}$.  By the assumption, we have $f(T) \le \hat{f}(T)$. Also, if we set $\alpha_T  = 1$ and $\alpha_S  = 0$ for rest of $S$ then we get $\hat{f}(T) \le f(T)$. Therefore,  $\hat{f}(T)$ is an extension. Now suppose for some $T,T_1,\dots,T_n \subseteq [m]$ and $\beta_1,\dots,\beta_n \ge 0$, we have  $\sum_{j=1}^{n} \beta_j \chi(T_j) \ge \chi(T)$. We will show that $\hat{f}(T) \le \sum_{j =1}^{n} \beta_j \hat{f}(T_j)$ which will complete the proof. Let $\hat{f}(T_i) = \sum_{S \in \mathcal{D}} \alpha^i_S f(S)$ for all $i \in [n]$. Here, $\{\alpha^i_S\}_{S \in \mathcal{D}}$ are optimal values as in definition of $\hat{f}$ and we have $\sum_{S \in \mathcal{D}} \alpha^i_S \chi(S) \ge \chi(T_i)$. Therefore, we need to show that $\hat{f}(T) \le \sum_{j =1}^{n} \beta_j \sum_{S \in \mathcal{D}} \alpha^j_S f(S) = \sum_{S \in \mathcal{D}}  (\sum_{j = 1}^{n} \beta_j \alpha^j_S) f(S)$. Note that $\chi(T) \le \sum_{j=1}^{n} \beta_j \chi(T_j) \le \sum_{j=1}^{n} \beta_j \sum_{S \in \mathcal{D}} \alpha^i_S \chi(S) = \sum_{S \in \mathcal{D}}  (\sum_{j = 1}^{n} \beta_j \alpha^j_S) \chi(S)$. Therefore, by definition of $\hat{f}$, we have  $\hat{f}(T) \le  \sum_{S \in \mathcal{D}}  (\sum_{j = 1}^{n} \beta_j \alpha^j_S) f(S)$.
  \end{proof}
 
 A set $T \in M_\lambda$ is called \textsf{bad} if it causes the tester to reject.  For subadditive functions,   the  set of bad sets $\mathcal{B}$ consists of $T \in M_\lambda$ such that either there exists  $T' \subseteq T$  such that   $f(T) < f(T')$ or there exists $T_1,\dots,T_r$ for some $r \ge 1$ such that $T = \cup_{i =1}^{r} T_i$ and $f(T) > \sum_{i =1}^{r} f(T_i)$. Similarly, for XOS functions, $T \in M_\lambda$ is in $\mathcal{B}$   if there exists $T' \subseteq T$ such that    $f(T) < f(T')$ or  there exists, for some $r \ge 1$,  $T_1,\dots,T_r \subseteq T$ such that $T = \cup_{i =1}^{r} T_i$ and $\alpha_1,\dots,\alpha_r \ge 0$   such that for all elements $s \in T$,  $\sum_{T_j:s \in T_j} \alpha_j \ge 1$ and   $ f(T) > \sum_{j =1}^{r} \alpha_j f(T_j)$. 
 
 We show that removing all sets not in $M_\lambda$, as well as the bad sets, gives us a partial function that can be extended to subadditive (or XOS). Since the function is $\epsilon$-far and $M_\lambda$ is large by our choice of $\lambda$, therefore there must be many bad sets.
 %Let $\mathcal{B}$ be the set all bad sets. 
 \begin{lemma}
 	The partial function $H = \{(S,f(S))| S \in M_\lambda \quad \text{and} \quad S \not \in \mathcal{B} \}$ is extendible to a  subadditive (XOS) function.
 \end{lemma}
 \begin{proof}
 	Suppose the partial function is not extendible, and let $\mathcal{D} =\{S| S \in M_\lambda \quad \text{and} \quad S \not \in \mathcal{B} \}$ be the defined points in $H$. Then for subadditive functions, by Lemma \ref{lemma-subadditive},  there  exist  $T_1,\dots,T_r,T  \in \mathcal{D}$ such that  $T \subseteq \cup_{i = 1}^{r} T_i$ and $f(T) > \sum_{i = 1}^{r} f(T_i)$. Then either $f(T) > \sum_{i = 1}^{r} f(T \cap T_i) $ or for some $j \in [r]$, $f(T_j) < f(T \cap T_j)$. 
 	Note that since $T,T_j \in M_\lambda$, so is $T \cap T_j$. Further, $T = \cup_{i \in [r]} T \cap T_i$. 
 	Thus, in the first case, $T$ is in $\mathcal{B}$ while in the second case, $T_j \in \mathcal{B}$, giving a contradiction.
 	
 	For XOS functions, by Lemma \ref{lemma-XOS} there exists  $T_1,\dots,T_r$ and $\alpha_1,\dots,\alpha_r \ge 0$ for some $r \ge 1$,  such that  $\sum_{T_j:s \in T_j} \alpha_j \ge 1$ for each $s \in T$ and $ f(T) > \sum_{j =1}^{r} \alpha_j f(T_j)$. Like subadditive functions, either $f(T) >  \sum_{j =1}^{r} \alpha_j f(T \cap T_j)$  or for some $j \in [r]$, $f(T_j) < f(T \cap T_j)$.  Again, in the first case, $T$ is in $\mathcal{B}$, while in the second case $T_j \in \mathcal{B}$. 
 \end{proof}
 \begin{theorem}
 	If $f$ is $\epsilon$-far from  subadditive (XOS) functions then the above tester rejects with constant probability.
 \end{theorem}
 \begin{proof}
 	Let $\mathcal{D} = \{S| S \in M_\lambda \quad \text{and} \quad S \not \in \mathcal{B}  \}$  and $\mathcal{U} = 2^{[m]}\setminus \mathcal{D}$. Since the partial function $\{(S,f(S))|S \in \mathcal{D}\}$ is extendible so $|\mathcal{U}| \ge \epsilon 2^m$ (since $f$ is $\epsilon$-far). 
 	Note that $|\mathcal{U}| = |\mathcal{B}| + \sum_{i = m/2 + \lambda \sqrt{m}}^{m} \binom{m}{i} = |\mathcal{B}| + \sum_{i = 1}^{m/2 - \lambda \sqrt{m}} \binom{m}{i}$. By Chernoff bound,  $\sum_{i = 1}^{m/2 - \lambda \sqrt{m}} \binom{m}{i} = 2^m  Pr(X \le m/2 - \lambda \sqrt{m}) \le 2^m e^{-\lambda^2}$ where  $X$ is a  binomial random variable $Bi(m,1/2)$. By the choice of $\lambda$,  $\sum_{i = 1}^{m/2 - \lambda \sqrt{m}} \binom{m}{i} \le \epsilon 2^m/4$. Hence we have $ |\mathcal{B}| \ge 3 \epsilon  2^m/4$.
 	Therefore in a single iteration our tester will pick a bad set with  probability at least $\frac{3}{4} \epsilon$. Hence after $1/\epsilon$ iterations, the tester will pick a bad set with constant probability. 
 	 
 \end{proof}

 \paragraph*{A subexponential tester for nonmonotone subadditive functions.} We now describe a property testing algorithm for general (nonmonotone) subadditive functions that makes $2^{O(\sqrt{m\log (1/\epsilon)}\log m)}$ queries; in this subsection, subadditive refers to nonmonotone subadditive functions.

 Let $\lambda = \sqrt{\ln (4/\epsilon)}$, and define $M_\lambda = \{S \subseteq [m] | m/2 - \lambda \sqrt{m} \le |S| \le m/2 + \lambda \sqrt{m}\}$.
  The tester repeats the following steps $1/\epsilon$ times:
 \begin{itemize}
 	\item Randomly pick a set $T \in M_\lambda$ and query the sets  $Q(T) = \{S \in M_\lambda | S \subseteq T \} $.
 	\item If there exists $T_1,\dots,T_r \in  Q(T)$ for some $r \ge 1$ such that $T = \cup_{i =1}^{r} T_i$ and $f(T) > \sum_{i =1}^{r} f(T_i)$ then reject.
 \end{itemize}
The tester makes $|Q|/\epsilon$ queries, where $|Q| \le O(\binom{m/2 + \lambda \sqrt{m}}{2\lambda \sqrt{m}})$, and hence $|Q| = 2^{O(\sqrt{m \log (1/\epsilon)} \log m)}$, which is also a bound on the number of queries by the tester.

 Obviously if the function $f$ is subadditive then the tester accepts.  Now we will show if the function $f$ is $\epsilon$-far from subadditive then tester rejects with high probability. 

We first give the characterization for partial function extension similar to claim \ref{lemma-subadditive} for general subadditive functions.
\begin{lemma}
	\label{lemma-general-subadditive}
	The partial function $H$ is extendible to a subadditive function (not necessarily monotone) iff  $\sum_{i = 1}^{r} f(T_i) \ge f(\cup_{i = 1}^{r} T_i)$ for all $T_1,\dots,T_r \in \mathcal{D}$ such that $\cup_{i = 1}^{r} T_i \in \mathcal{D}$.
\end{lemma}

A set $T \in M_\lambda$ is called \textsf{bad} if it causes the tester to reject.  The  set of bad sets $\mathcal{B}$ consists of $T \in M_\lambda$ such that there exists $T_1,\dots,T_r \in  M_\lambda$ for some $r \ge 1$ such that $T = \cup_{i =1}^{r} T_i$ and $f(T) > \sum_{i =1}^{r} f(T_i)$.  

 We show that removing all sets not in $M_\lambda$, as well as the bad sets, gives us a partial function that can be extended to subadditive function. Since the function is $\epsilon$-far and $M_\lambda$ is large by our choice of $\lambda$, therefore there must be many bad sets.
\begin{lemma}
	The partial function $H = \{(S,f(S))| S \in M_\lambda \quad \text{and} \quad S \not \in \mathcal{B} \}$ is extendible to a subadditive function.
\end{lemma}
\begin{proof}
	Suppose the partial function is not extendible. Let $\mathcal{D} =\{S| S \in M_\lambda \quad \text{and} \quad S \not \in \mathcal{B} \}$ be the defined sets in $H$. Then by Lemma \ref{lemma-general-subadditive}, there will exist  $T_1,\dots,T_r,T  \in \mathcal{D}$ such that  $T = \cup_{i = 1}^{r} T_i$ and $\sum_{i = 1}^{r} f(T_i) < f(T)$. This implies $T \in \mathcal{B}$ which is a contradiction. 
\end{proof}
\begin{theorem}
	If $f$ is $\epsilon$-far from subadditive functions then the above tester rejects with constant probability.
\end{theorem}
\begin{proof}
	Let $\mathcal{D} = \{S| S \in M_\lambda \quad \text{and} \quad S \not \in \mathcal{B} \}$  and $\mathcal{U} = 2^{[m]}\setminus \mathcal{D}$. Since the partial function $\{(S,f(S))|S \in \mathcal{D}\}$ is extendible,  $|\mathcal{U}| \ge \epsilon 2^m$ (since $f$ is $\epsilon$-far).
	Note that $|\mathcal{U}| = |\mathcal{B}| +2 \sum_{i = 1}^{m/2 - \lambda \sqrt{m}} \binom{m}{i}$. Hence again using Chernoff bound and the value of $\lambda$, we have $ |\mathcal{B}| \ge \epsilon 2^m/2$.  Therefore in a single iteration our tester will pick a bad set with $\epsilon/2$ probability. Hence after $1/\epsilon$ iterations, the tester will pick a bad set with constant probability.
\end{proof} 
		\section{Submodular functions}

Submodular functions are perhaps the the most important functions in combinatorial optimization. They capture diminishing marginal returns for set functions, which is satisfied in many practical applications as well as theoretical constructs, and can be viewed as the discrete analog of concave functions. Coverage functions, matroid rank functions,  and many others are special cases of submodular functions. A function $f: 2^{[m]} \rightarrow \mathbb{R}$ is submodular if $f(A) + f(B) \ge f(A \cup B) + f(A \cap B)$ for all $A,B \subseteq [m]$. 
As before, $\mcD := \{T_i\}_{i \in [n]}$ is the set of points in the given partial functions $H$, and $\mcU := 2^{[m]}$. These are called \emph{defined} and \emph{undefined} points respectively. The \emph{lattice closure} $LC(\mcD)$ is the minimal set that contains $\mcD$ and is closed under union and intersection. For a family of set $\mcF \subseteq 2^{[m]}$, we say a function $f$ is submodular on $\mcF$ if for all sets $A$, $B \in \mcF$ so that $A \cup B$, $A \cap B$ are also in $\mcF$, $f(A) + f(B) \ge f(A \cup B) + f(A \cap B)$.

		Seshadri and Vondrak \cite{submodularity} study  property testing of submodular functions and give the first subexponential time non-adaptive tester for submodularity. They also introduce the problem of Submodular Extension --- the problem we study --- to a submodular function, and note its usefulness in analyzing  property testing algorithms. They give a partial function $H$ that is defined on an exponential number of points and is not extendible to a submodular function. However, if any  point is removed from $H$, then it can be extended to a submodular function, indicating the difficulty in designing a tester for submodularity. %Like before, we study the problem from computational perspective.

Seshadri and Vondrak introduce and use a combinatorial structure called a \emph{path certificate}, the existence of which certifies that a given partial function cannot be extended to a submodular function. Our results are instead based on a structure called a \emph{square certificate}. A square certificate is multiset of tuples $(\{A,B\},A \cup B,A \cap B)$ with $A, B \subseteq [m]$, with additional constraints on the multiset. Since the union and intersection of sets can be seen as bitwise AND and OR of the characteristic vectors of the sets, our analysis proceeds by viewing a square certificate as special monotone Boolean circuit. Since submodular functions are defined using union and intersection, we believe that square certificates are more natural than path certificates. It is easy to show that a square certificate exists iff the given partial function is not extendible. Our main technical result is that if a partial function is not extendible, then there exists a square certificate where the sets in every square $(\{A,B\},A \cup B,A \cap B)$  are in the lattice closure of $\mcD$. We use this to obtain a number of results on the extendibility of a partial function.

	\begin{theorem}
		\label{submodular-main1}
		\begin{enumerate}
			\item If the sets in $\mcD$ form an antichain,\footnote{A family of sets is an antichain if no set in the family is contained in another set.} then the partial function can always be extended to a submodular function. Hence, submodular functions cannot be PMAC-learned.
			\item Let $\mathcal{F} :=  LC(\mcD) \cap \{S \, : \, \exists T_i, T_j \in \mcD \mbox{ s.t. } T_i \supseteq S \supseteq T_j \}$ be the sets obtained by the union and intersection of sets in $\mcD$, that are also both contained in and contained by sets in $\mcD$. If the partial function can be extended to a submodular function on $\mathcal{F}$, then it can be extended to a submodular function on $2^{[m]}$. Hence, Submodular Approximate Extension (and thus Extension) can be solved in $O(\poly(|\mathcal{F}|,m,n))$ time.
		\end{enumerate}
	\end{theorem}

Let $r$ be the maximum difference in the size of two sets $T_i, T_j \in \mcD$. Then an upper bound on the size of $\mcF$ is $O(\poly(n,2^r))$. Thus if all sets are roughly of the same size, Approximate Extension can be solved in polynomial time. Further, if $n$ is a constant then $|LC(\mcD)|$ is a constant and hence Approximate Extension can be solved in polynomial time. For the proof of the theorem, the second result is in particular very non-trivial and require insights into the structure of square certificates. We start by formally defining square certificates.

\vspace{0.1in}
\noindent \textbf{Square Certificates.} A \emph{square tuple} is a triple $(\{A,B\}, A \cup B,  A \cap B)$ where $A,B \subseteq [m]$. The sets $A$ and $B$ are called \emph{middle} points,  $A \cup B$ is the \emph{top} point and $A \cap B$ is the \emph{bottom} point of this square tuple. Sets $A,B,A \cup B, A \cap B$ are said to be \emph{part} of this square tuple. Given a multiset of square tuples, for a set $S \subseteq [m]$, define $m(S)$ to be the number of square tuples with $S$ as a middle point, and $tb(S)$ to be the number of square tuples with $S$ as a top or bottom point. We say a set $S$ is an \emph{input set} if $m(S) > tb(S)$, an \emph{intermediate set} if $m(S) = tb(S)$, and an \emph{output set} if $m(S) < tb(S)$. A \emph{square certificate} is a multiset of square tuples with the following properties:

\begin{enumerate}[(P1)]
	\item If $S$ is an input or an output set, i.e., $m(S) \neq tb(S)$then $S$ must be in $\mcD$.
	\item $\sum_{i \in [n]} f_i \left(tb(T_i) - m(T_i)\right) > 0$.
\end{enumerate}

\noindent  A set $S$ is \emph{involved} if it is a part of some square triple in the square certificate. By definition, a partial function can be extended iff the following linear program (with variables $w_A$ for all $A \subseteq [m]$) is feasible: 
 
 \[
 w_A + w_B \ge w_{A \cup B} + w_{A \cap B} \quad \forall A,B \subseteq [m], \quad w_A = f_A \quad \forall A \in \mathcal{D} \,.
 \]

\noindent A square certificate can be seen as a dual solution obtained from Farkas' lemma. 
\begin{lemma}
	\label{certificate-defn}
	A partial function is extendible iff there does not exist a square certificate.
\end{lemma}
\begin{proof}
	 A partial function can be extended iff the following linear program (with variables $w_S$ for all $S \subseteq [m]$) is feasible: 
	 
	 \[
	 w_A + w_B \ge w_{A \cup B} + w_{A \cap B} \quad \forall A,B \subseteq [m], \quad w_A = f_A \quad \forall A \in \mathcal{D} \,.
	 \]
	 
	 \noindent Using Farkas' lemma, the following linear program with variables $y_{\{A,B\}}$ for all $A,B \subseteq [m]$ must then be infeasible: 
	 \begin{align}
	 \sum_{B} y_{\{A,B\}} - \sum_{P,Q: P \cup Q = A} y_{\{P,Q\}} - \sum_{P,Q: P \cap Q = A} y_{\{P,Q\}} & ~ = ~ 0 \quad \forall \thinspace A \not \in \mathcal{D} \\
	 \sum_{A \in \mathcal{D}}f_A \Big( \sum_{P,Q: P \cup Q = A} y_{\{P,Q\}} + \sum_{P,Q: P \cap Q = A} y_{\{P,Q\}} - \sum_{B} y_{\{A,B\}} \Big)& ~>~ 0 \\
	 y_{\{A,B\}} & ~\ge~ 0 \quad \forall A,B \subseteq [m]
	 \end{align}
	 
	 \noindent We will show that there exists a square certificate iff this dual linear program is feasible. Firstly, note that if the dual is feasible then it has a rational solution, and any rational solution to the dual can be converted into an integral solution by multiplying by the product of the denominators. Now given a square certificate, we identify each variable $y_{\{A,B\}}$ in the dual with the square tuple $(\{A,B\},A \cup B, A \cap B)$, and set $y_{\{A,B\}}= k$ if this square tuple appears $k$ times in the square certificate. Then for a set $A$, $m(A) = \sum_B y_{\{A,B\}}$, and $tb(A)$ $= \sum_{P,Q: P \cup Q = A} y_{\{P,Q\}}$ $+ \sum_{P,Q: P \cap Q = A} y_{\{P,Q\}}$. The first constraint then says that if $A$ is not defined, then $A$ must appear an equal number of times as a middle point, and as a top or bottom point. That is, $A$ must be an intermediate set. The second constraint in the linear program corresponds exactly to property (P2). Thus, a square certificate gives us a feasible dual solution. The converse can be proved with the same construction: given an integral dual solution, we create a square certificate by including a square tuple $(\{A,B\}, A \cup B, A \cap B)$ $y_{\{A,B\}}$ times. The first and second constraints in the linear program give us properties (P1) and (P2) of the square certificate exactly.
	\end{proof}

\begin{lemma}
	\label{vs-in-sc}
If there is  a square certificate for a partial function, then for any extension $f(\cdot)$ to the hypercube, there is a square tuple $(A,B,A\cup B, A \cap B)$ in the square certificate with $ f(A)  + f(B) <  f(A \cup B) + f(A \cap B)$.
\end{lemma}

Our main technical result is the following lemma.

\begin{lemma}
If there is a square certificate for a partial function, then there is a square certificate where, if $S$ is an involved set, then (i) there exist $T_i$, $T_j \in \mcD$ so that $T_i \supseteq S \supseteq T_j$, and (ii) $S \in LC(\mcD)$.
\label{submodular-lemma}
\end{lemma}

We will develop the proof of the second part of the lemma in the remainder of the paper. The first part, however, is easily seen. Let $S$ be an involved set. Then either by property (P1) $S \in \mcD$, in which case $T_i = T_j = S$, or $S$ is undefined and hence $S$ is an intermediate set. In the latter case, $m(S) = tb(S) > 0$, and hence $S$ is a middle set in some square $(S,S', S \cup S', S \cap S')$. If $S \cup S'$ is defined, then $T_i = S \cup S'$. Otherwise, $S \cup S'$ is an intermediate set, and we can continue in this way until we find a defined set that contains $S$. Similarly, considering the bottom point $S \cap S'$ of the square, we can find a defined set that is contained by $S$.

\begin{proof}[Proof of Theorem~\ref{submodular-main1} (assuming Lemma \ref{submodular-lemma}).]
For the first part of the theorem, if the partial function cannot be extended, there must exist a square certificate. In particular, from (P2), there must exist $T_i \in \mcD$ with $tb(T_i) - m(T_i) > 0$. Thus, $T_i$ must be part of at least one square tuple $(A,B, A\cup B, A \cap B)$ that is not trivial, i.e., $A \neq B$ (otherwise all four sets are just equal to $T_i$, and this square contributes to $tb(T_i)$ and $m(T_i)$ equally). Then $A \cap B \neq A \cup B$, and either $T_i \subsetneq A \cup B$, or $T_i \supsetneq A \cap B$ (or both). In the former case, by Lemma~\ref{submodular-lemma}, there must exist $T_j \supseteq A \cup B \neq T_i$ with $T_j \in \mcD$, or in the latter case, there must exist $T_j \subseteq A \cap B \neq T_i$ with $T_j \in \mcD$. However, since $\mcD$ is an antichain, no such sets $T_i$, $T_j$ can exist.

To prove that submodular functions cannot be learned, consider the family $\mathcal{F} = {[m] \choose \frac{m}{2}}$ of sets of size $m/2$, and note that $|\mathcal{F}| = \Theta(2^m)$.  Since any partial function with sets in $\mcF$ and arbitrary values for the sets is extendible to a submodular function,  submodular functions can not be learned by any factor (by Lemma~\ref{learning-application}).

For the second part of the theorem, let $f(\cdot)$ be a submodular extension of the partial function to $\mcF$. Then for any points $A, B \in \mcF$ so that $A\cup B, A\cap B$ are also in $\mcF$, $f(A) + f(B) \ge f(A \cup B) + f(A \cap B)$. Suppose for a contradiction that the partial function cannot be extended to a submodular function on the entire domain. Then by Lemma~\ref{certificate-defn} and Lemma~\ref{submodular-lemma}, there is a square certificate where all involved sets are in $\mcF$. But then we have a square certificate and an extension $f(\cdot)$ for which $f(A) + f(B) \ge f(A \cup B) + f(A \cap B)$ for all square tuples in the square certificate. This gives us a contradiction, since by Lemma~\ref{vs-in-sc} the two cannot coexist. Thus, it is necessary and sufficient to find a submodular extension to points in $\mcF$. Therefore, Approximate Extension can be solved by writing a linear program  with the variable $w_S$ for each set $S \in \mcF$ and objective of minimizing $\alpha \ge 1$ subject to  submodularity constraints on the sets $w_S + w_T \ge w_{S \cup T} + w_{S \cap T}$ for $S,T, S \cup T, S \cap T \in \mcF$, and the constraints $f_i \le w_{T_i} \le \alpha f_i$ for each set $T_i \in \mcD$. This gives an algorithm with the desired running time.
\end{proof}

 Given a square certificate, we need to construct another square certificate where all involved sets are in the lattice closure of $\mcD$, i.e., are obtained by the union and intersection of sets in $\mcD$. In particular, this means that all intermediate sets must be in $LC(\mcD)$, since by property (P1) input and output sets are already in $\mcD$. Our insight is to think of the union and intersection of two sets as the bitwise OR and AND of the characteristic vectors respectively. Given a square certificate, we will construct a boolean circuit with certain properties, which we call a \emph{boolean certificate}. Similar to square certificates, boolean certificates are also necessary and sufficient conditions for a partial function to not be extendible to a submodular function. 

As background for defining boolean certificates, we first define boolean circuits. Formally, a \emph{boolean circuit} is a directed graph with three types of nodes --- \emph{input} nodes $\ip$ with no incoming edges, \emph{intermediate} nodes $\im$ with incoming and outgoing edges, and \emph{output} nodes $\op$ with no outgoing edges. Intermediate and output nodes are labelled with logic gates and perform the labelled boolean operation. The computed value at an intermediate node is fed to other intermediate or output nodes, and the computed value at the output nodes are the output of the circuit. The fan-in of a gate is its indegree, which is the number of arguments for the boolean operation. Throughout the paper, we will consider only AND and OR gates with fan-in 2. An assignment $A:\ip \cup \op \cup \im \rightarrow \{0,1\}$ is called \emph{satisfying} if the boolean operation operation at each gate is satisfied by the assignment. That is, if $g$ is an AND gate with inputs $g_1$ and $g_2$, then $A(g) = A(g_1) \land A(g_2)$, and $A(g) = A(g_1) \lor A(g_2)$ if $g$ is an OR gate. 

%Boolean circuits are thus assumed to be acyclic as a cycle implies two  nodes depend on each other to compute their values. 
\begin{figure}
	\begin{tikzpicture}[scale = .85]
	\node[style = {circle,thick,draw},label = right :{$y_1$}]  (a) at (0, 0) {\small $\land$};
	\node[style = {circle,thick,draw}
	,label = right :{$z_1$}
	]  (b) at ( 1, 1) {\small $\land$};
	\node[style = {circle,thick,draw}
	%,label = below :{$\begin{bmatrix} 1 \\ 0 \\0 \end{bmatrix}$}
	]  (c) at ( -1, 1) {$x_1$};
	\node[style = {circle,thick,draw}
	,label = above :{$z_2$}
	]  (d) at ( 0, 2) {\small $\lor$};
	\node[style = {circle,thick,draw}
	%,label = below :{$\begin{bmatrix} 0 \\ 1 \\1 \end{bmatrix}$}
	]  (e) at ( 2.5, 3) {$x_2$};
	\node[style = {circle,thick,draw}
	,label = right :{$y_2$}
	]  (f) at ( 1.5, 4) {\small $\lor$};
	\draw[->,>=stealth] (b) -- (a);
	\draw[->,>=stealth] (c) -- (a);
	\draw[->,>=stealth] (b) -- (d);
	\draw[->,>=stealth](c) -- (d);
	\draw [->,>=stealth] (d) to [out=30,in=100] (b);
	\draw[->,>=stealth] (e) -- (b);
	\draw[->,>=stealth] (d) -- (f);
	\draw[->,>=stealth] (e) -- (f);	
	\end{tikzpicture}	
	\caption{A cycle between $\land$ and $\lor$ gates. For input values $x_1 = 0, x_2 = 1$, both $z_1 = z_2 = y_1 = 0, y_2 = 1$ and $z_1 = z_2 = y_2 = 1, y_1 = 0$ are satisfying assignments. }
	\label{fig:cycle}
\end{figure}
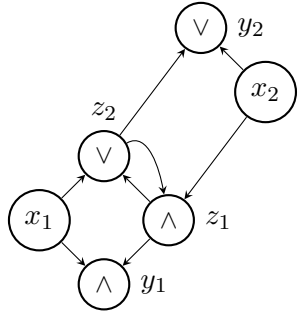
The boolean circuits we consider are constructed from  square certificates and can contain cycles. In a cyclic boolean circuit,  given values to the input nodes, there may be many satisfying assignments (see figure \ref{fig:cycle}).   The technical difficulty in proving second part of Lemma~\ref{submodular-lemma} is due to the presence of cycles in the constructed boolean circuit (the square certificate may correspondingly have analogous ``cycles'').  If the constructed boolean circuit does not have cycles, then the proof of the lemma follows trivially.

%Before we discuss boolean certificates,
 We now define the notion of computation in cyclic boolean circuits. We use $\all$ to denote the set of all gates.% The cardinality of $X,Y,Z$ is assumed to be $n_1,n_2$ and $n_3$ respectively. 
  We assume an ordering of the gates, with the input gates appearing first in this order, followed by the output gates and then the intermediate gates. An assignment to the input gates is denoted $(x_1,\dots,x_{|\ip|}) \in \{0,1\}^{|\ip|}$. An assignment to all gates will be denoted by $(x_1,\dots,x_{|\ip|},y_1,\dots,y_{|\op|},z_1,\dots,z_{|\im|}) \in \{0,1\}^{|\all|}$. %An assignment to all gates extends an assignment to input gates if they have same value on each input gate. 
We use $X_1$, $X_2$, $\ldots$, $X_{|\ip|}$ to refer to the input gates.

%We assume only AND and OR gates in this paper.

Recall that an assignment $(x_1,\dots,x_{|\ip|},y_1,\dots,y_{|\op|},z_1,\dots,z_{|\im|})$ is called \emph{satisfying} if the boolean operation operation at each gate is satisfied by the assignment. Given an assignment $(x_1,\dots,x_{|\ip|})$ to input gates,  we say a gate $g^*$ \emph{computes} a value $b^* \in \{0,1\}$ if the value at $g^*$ is $b^*$ in every satisfying assignment to the gates, with the input gates as specified.  We now formally define the notion of a gate being \emph{fixed} to a value $b \in \{0,1\}$. 
The equivalence of the two definitions is shown in Lemma~\ref{intutive-defn}.  It is possible that a gate does not compute any value.

 \begin{definition}
 	Given an assignment $(x_1,\dots,x_{|\ip|}) \in \{0,1\}^{|\ip|}$ to the  input gates, a gate $g^*$ is  \emph{fixed} to value $b^* \in \{0,1\}$ if there exists a subgraph $G' = (V' , E')$ of the boolean circuit with values $\val(g) \in \{0,1\}$ for all gates $g \in V'$, so that:
 	\begin{itemize}
 		\item $G'$ is a rooted tree with all edges directed towards the root.
 		\item $g^*$ is the root, and all leaves are input nodes.
 		\item $\val(g^*) = b^*$, and $\val(X_i) = x_i$ for each input node $X_i$ that is a leaf.
 		\item If gate $g$ has only one child $g'$ in the tree then either (i) $g$ is an AND gate and $\val(g) = \val(g') = 0$ or (ii) $g$ is an OR gate and $\val(g) = \val(g') = 1$.  
 		%$g_2$ such that $(g_1,0) \in \mathcal{F}$  and $g$ is AND gate taking inputs from $g_1$ and $g_2$ then $(g,0) \in \mathcal{F}$.
 		\item If gate $g$ has two children $g'$, $g''$ in the tree then either (i) $g$ is an AND gate and $\val(g) = \val(g') = \val(g'') = 1$ or (ii) $g$ is an OR gate and $\val(g) = \val(g') = \val(g'') = 0$. 
 		%\item If there exists gate $g_1$ and $g_2$ such that  $(g_1,1) \in \mathcal{F}$  and $g$ is OR gate taking inputs from $g_1$ and $g_2$ then $(g,1) \in \mathcal{F}$.
 		%\item If there exists gate $g_1$ and $g_2$ such that $(g_1,b_1)$ and $(g_2,b_2)$ are  in  $\mathcal{F}$ and $g$ is AND gate taking inputs from $g_1$ and $g_2$ then $(g,b_1 b_2) \in \mathcal{F}$.
 		%\item If there exists gate $g_1$ and $g_2$ such that $(g_1,b_1)$ and $(g_2,b_2)$ are  in  $\mathcal{F}$ and $g$ is OR gate taking inputs from $g_1$ and $g_2$ then $(g,b_1 + b_2) \in \mathcal{F}$.
 	\end{itemize}
 	\label{def:fixing}
 \end{definition}
 We  call the above rooted tree and associated values a \emph{proof} of $g^*$ fixed to $b^*$ for the given assignment to the input gates. We make a few observations regarding the definition. First, since in the boolean circuit the fan-in for each gate is 2, the rooted tree is a binary tree. Second, since the value of every gate in the rooted tree must be the same as its children, in fact every gate in the rooted tree must have the same value.
 
 Note that a priori, a gate may be fixed to different values for the same input. We now show that in fact this cannot happen. 
 \begin{lemma}
 	\label{fixed-value}
 	If a gate $g$ is fixed to both $b'$ and $b''$ for a particular assignment to the input gates, then $b' = b''$.
 \end{lemma}

 Define $\fix^{(x_1,\dots,x_{|\ip|})}$ to be the set of all gates that gets fixed by the assignment $(x_1,\dots,x_{|\ip|})$ to the input gates. We say such gates are \emph{fixed}, with the inputs clear from the context. If $g$ is fixed, then Lemma \ref{fixed-value} allows us to define $\fix^{(x_1,\dots,x_{|\ip|})}(g)$ as the value that gate $g$ is fixed to.
 
 Lemma \ref{intutive-defn} says that the value of fixed gates should be consistent with any satisfying assignment and Lemma \ref{monotone} says that $\fix^{(x_1,\dots,x_{|\ip|})}(g)$ is monotone for fixed gate $g$ as a function of $(x_1,\dots,x_{|\ip|})$. These lemmas are intuitive and are proven by induction on the height of the rooted tree corresponding to any proof that fixes the gate $g$.
 
 \begin{lemma}
 	\label{intutive-defn}
 	Let $A : \ip \cup \op \cup \im \rightarrow \{0,1\}$ be a satisfying assignment to a boolean circuit, and gate $g$ be fixed by the assignment to the input gates. Then $g$ is fixed to $A(g)$. 
 \end{lemma}

 For two vectors $v$, $v'$ we use the standard notation that $v \ge v'$ if this is true of each component.
 
 \begin{lemma}
 	\label{monotone}
 	Let $(x_1, \ldots, x_{|\ip|})  \ge (x_1', \ldots, x_{|\ip|}')$ be two assignments to the input gates, and consider a gate $g$. If $g$ is fixed to $0$ by the assignment $(x_1, \ldots, x_{|\ip|})$, then this is true for $(x_1', \ldots, x_{|\ip|}')$ as well. Conversely, if $g$ is fixed to $1$ by the assignment $(x_1', \ldots, x_{|\ip|}')$, then this is true for $(x_1, \ldots, x_{|\ip|})$ as well.
 \end{lemma}

\paragraph*{Boolean Certificates.} Given a partial function $H = \{(T_1, f_1), \ldots, (T_n,f_n)\}$, a boolean certificate, similar to a square certificate, characterizes partial functions that cannot be extended to a submodular function. A boolean certificate consists of the following two parts. 

\begin{enumerate} 
	\item A boolean circuit with AND and OR gates, that satisfies two conditions: (i) Input gates and intermediate gates have outdegree two, and output and intermediate gates have indegree two. (ii) If $g$, $g'$ are inputs to an AND gate, then these are inputs to an OR gate as well. As before, $\ip$, $\op$, and $\im$ are the set of input, output, and intermediate gates, and $\all$ is the set of all gates. 
	\item A function $\mcC: \all \rightarrow 2^{[m]}$ called the creator function, that satisfies three conditions: (i) $\mcC(g) \in \mcD$ for all $g \in \ip \cup \op$, i.e., all input and output gates map to defined sets. (ii) For each $i \in [m]$, the assignment  $A_i(g) = (\mcC(g))_i$ (which assigns $1$ to the gate if the $i$th element is present in $\mcC(g)$, and $0$ otherwise) is a satisfying assignment. This gives $m$ satisfying assigments to the boolean circuit.
	(iii) For $T_i \in \mcD$, let $n_i^{\ip}$ be the number of input gates that $\mcC$ maps to $T_i$, and $n_i^{\op}$ be the number of output gates that $\mcC$ maps to $T_i$. Then $\sum_{i \in [n]} f_i \left( n_i^{\op} - n_i^{\ip} \right) > 0$.
\end{enumerate}

\begin{figure}[!ht]
	\begin{tikzpicture}[scale = 1]
	\node[style = {circle,thick,draw}]  (a) at ( 0, 0) {$Z$}; 
	\node[style = {circle,thick,draw}] (b) at ( 4,0) {$X_1$};
	\node[style = {circle,thick,draw}] (c) at ( 1,1) {$Y_1$}; 
	\node[style = {circle,thick,draw}] (d) at ( 1,-1) {$Y_2$};
	\node[style = {circle,thick,draw}]  (e) at (-0.5,2) {$X_2$}; 
	\node[style = {circle,thick,draw}] (f) at ( 0.5,2) {$X_3$};
	\node[style = {circle,thick,draw}] (g) at ( 0,3) {$Y_3$}; 
	\node[style = {circle,thick,draw}] (h) at ( -1,1) {$Y_4$};
	\node[style = {circle,thick,draw}] (i) at ( -4,0) {$X_4$};
	\node[style = {circle,thick,draw}] (j) at ( -1,-1) {$Y_5$};
	\node[style = {circle,thick,draw}] (k) at ( -0.5,-2) {$X_5$};
	\node[style = {circle,thick,draw}] (l) at ( 0.5,-2) {$X_6$};
	\node[style = {circle,thick,draw}] (m) at ( 0,-3) {$Y_6$};
	
	\draw (a) -- (c);
	\draw (b) -- (c);
	\draw (a) -- (d);
	\draw (b) -- (d);
	\draw (a) -- (h);
	\draw (a) -- (j);
	\draw (i) -- (h);
	\draw (i) -- (j);
	\draw (e) -- (a);
	\draw (f) -- (a);
	\draw  (f) -- (g);
	\draw  (e) -- (g);
	\draw  (k) -- (a);
	\draw  (l) -- (a);
	\draw  (k) -- (a);
	\draw  (l) -- (m);
	\draw  (k) -- (m);
	\end{tikzpicture}
	\begin{tikzpicture}[scale = 1]
	\node[style = {circle,thick,draw},label = $Z$]  (a) at ( 1, 0) {$\lor$}; 
	\node[style = {circle,thick,draw},label = $X_1$] (b) at ( 5,0) {$x_1$};
	\node[style = {circle,thick,draw},label = $Y_1$] (c) at ( 2,1) {$\lor$}; 
	\node[style = {circle,thick,draw},label = $Y_2$] (d) at ( 2,-1) {$\land$};
	
	\node[style = {circle,thick,draw},label = left: $X_5$] (k) at ( 0.5,-2) {$x_5$};
	\node[style = {circle,thick,draw},label =right: $X_6$] (l) at ( 1.5,-2) {$x_6$};
	\node[style = {circle,thick,draw},label = below: $Y_6$] (m) at ( 1,-3) {$\land$};

	\node[style = {circle,thick,draw},label = right:$Z$]  (a1) at ( -1, 0) {$\land$};
	
	\node[style = {circle,thick,draw},label = left:$X_2$]  (e) at (-1.5,2) {$x_2$}; 
	\node[style = {circle,thick,draw},label = right:$X_3$] (f) at ( -0.5,2) {$x_3$};
	%		\node[style = {circle,thick,draw},label = above:$g_{34}(Y_3)$] (g) at ( -1,3) {$\lor$}; 
	\node[style = {circle,thick,draw},label = above:$Y_3$] (g) at ( -1,3) {$\lor$}; 
	\node[style = {circle,thick,draw},label = left:$Y_4$] (h) at ( -2,1) {$\lor$};
	\node[style = {circle,thick,draw},label = above:$X_4$] (i) at ( -5,0) {$x_4$};
	\node[style = {circle,thick,draw},label = right:$Y_5$] (j) at ( -2,-1) {$\land$};

	\draw[->,>=stealth] (a) -- (c);
	\draw[->,>=stealth](b) -- (c);
	\draw[->,>=stealth] (a) -- (d);
	\draw[->,>=stealth] (b) -- (d);
	\draw[->,>=stealth] (a1) -- (h);
	
	\draw[->,>=stealth] (a1) -- (j);
	\draw[->,>=stealth] (i) -- (h);
	\draw[->,>=stealth] (i) -- (j);
	\draw[->,>=stealth] (e) -- (a1);
	\draw[->,>=stealth] (f) -- (a1);
	\draw[->,>=stealth]  (f) -- (g);
	\draw[->,>=stealth]  (e) -- (g);
	\draw[->,>=stealth]  (k) -- (a);
	%		\draw[->,>=stealth] (k1) -- (g1);
	\draw[->,>=stealth] (l) -- (a);
	\draw[->,>=stealth]  (k) -- (a);
	\draw[->,>=stealth]  (l) -- (m);
	\draw[->,>=stealth]  (k) -- (m);
	\end{tikzpicture}
	\label{fig:sc}
	\caption{A square certificate with input sets $\{X_1,\dots,X_6\}$, output sets $\{Y_1,\dots,Y_6\}$ and intermediate set $\{Z\}$, and the corresponding boolean certificate with the creator function values shown outside the circles.	}
	%			A  matching $M$: square $(\{Z,X_1\},Y_1,Y_2,1)$ and $(\{X_5,X_6\},Z,Y_6,2)$ are matched at $Z$;  $(\{X_4,Z\},Y_4,Y_5,5)$ and $(\{X_2,X_3\},Y_3,Z,3)$ are matched at $Z$;  $(\{Y_3,X_8\},Y_9,Y_8)$ and $(\{X_2,X_3\},Y_3,Z,3)$ are matched at $Y_3$}
\end{figure}
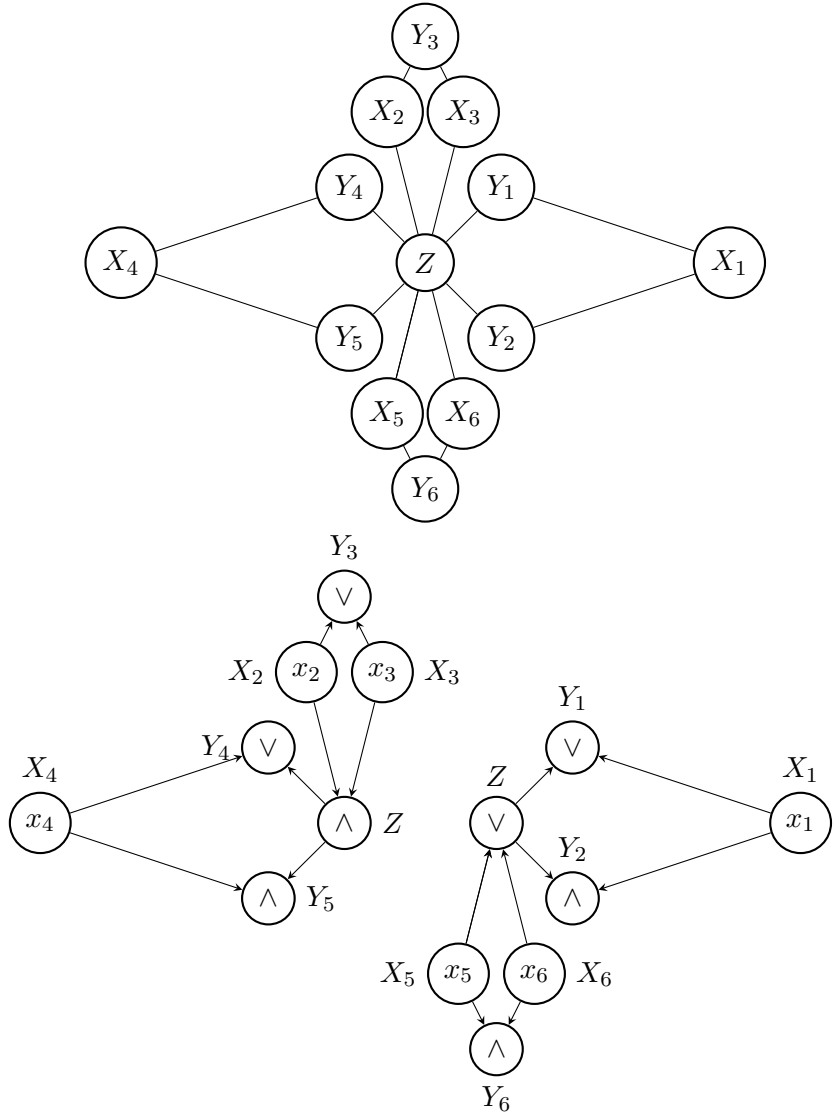

Figure 2 shows an example of a square certificate and the corresponding boolean certificate.

\begin{lemma}
	Given a partial function, there exists a square certificate iff there exists a boolean certificate. If there is a boolean certificate with creator function $\mcC$ then there is a square certificate  with  $\{\mcC(g)| g \in \all\}$ as the family of involved sets.
	\label{lem:squareboolean}
\end{lemma}

\begin{proof}
	Given a square certificate, we first show how to obtain a boolean certificate. The construction proceeds in two steps. In the first step, for each square tuple  $(\{A,B\}, A \cup B, A \cap B)$ in the square certificate, we create four gates $g_1$, $g_2$, $g_3$, $g_4$, and set $\mcC(g_1) = A$, $\mcC(g_2) = B$, $\mcC(g_3) = A \cup B$, and $\mcC(g_4) = A \cap B$. Gates $g_1$, $g_2$ are input gates, $g_3$ is an output OR gate, and $g_4$ is an output AND gate. We add edges from both $g_1$ and $g_2$ to $g_3$ and $g_4$. We do this for each square tuple in the square certificate. Thus, if there are $N$ square tuples, we obtain a boolean circuit with $N$ components and $4N$ gates. We call this the primary boolean circuit, and note its properties below.
	
	\begin{enumerate}
		\item The circuit only has input gates and output gates. Each input gate has fan-out 2, and each output gate has fan-in 2. Further, if $g$, $g'$ are inputs to an AND gate, these are inputs to an OR gate as well.
		\item Consider an involved set $S$ in the square certificate. For every square tuple that has $S$ as a top or bottom point, there is a component in the circuit that has $g$ as an output gate, and $\mcC(g) = S$. For every square tuple that has $S$ is a middle point, there is a square circuit that has $g$ as an input gate, and $\mcC(g) = S$. The converse also holds. Thus, if $n_S^{\ip}$ and $n_S^{\op}$ are the number of input and output gates respectively that the creator function maps to $S$, then $n_S^{\op} - n_S^{\ip} = tb(S) - m(S)$. Hence, $\sum_{i \in [n]} f_i \left( n_i^{\op} - n_i^{\ip} \right) > 0$.
		\item For each $i \in [m]$, the assignment $A_i(g) = (\mcC(g))_i$ is a satisfying assignment. This can be verified by checking each component.
		\item However, $g$ may be an input or output gate, but $\mcC(g)$ an intermediate set, and hence not in $\mcD$.
	\end{enumerate}
	
	Hence, the primary boolean circuit (with the creator function) satisfies all the conditions to be a boolean certificate, except that an input or output gate may be mapped to a set not in $\mcD$. In the second step, we fix this. Now we take two gates $g_1$, $g_2$ of which one (say $g_1$) is an input gate, one ($g_2$) is an output gate, and $\mcC(g_1) = \mcC(g_2)$. We replace these two by a single intermediate gate $g_3$ that takes inputs as did $g_2$, provides output as did $g_1$, has logical operation (AND or OR) as did $g_2$, and set $\mcC(g_3) = \mcC(g_1)$. Note that $g_3$ has fan-in and fan-out 2. We do this for all pairs of gates that satisfy these conditions, one pair at a time. We call the circuit thus obtained the secondary boolean circuit. It is easy to verify by induction on the number of gates replaced that the properties of a boolean certificate satisfied earlier are satisfied now as well. Finally, consider a gate $g \in \ip$, i.e., an input gate in the secondary boolean circuit. Let $\mcC(g) = S$. Since $g$ has not been replaced, there is no output gate that the creator function maps to set $S$. Then since $n_S^{\op} - n_S^{\ip} = tb(S) - m(S)$, it must be true that $m(S) > tb(S)$, and hence $S \in \mcD$ as required. Similarly, if $g$ is an output gate in the secondary boolean circuit, then $\mcC(g) \in \mcD$. This concludes the construction of the boolean certificate.
	
	Given a boolean certificate, we now want to construct a square certificate so that the involved sets are exactly $\{\mcC(g)| g \in \all\}$. For this, we first convert the boolean circuit to a primary boolean circuit by reversing the earlier procedure. That is, if $g$ is an intermediate gate (and hence has two inputs and two outputs), we replace it by an input gate $g_1$ that provides output as did $g$ , and an output gate $g_2$ that takes inputs as did $g$, with $\mcC(g_1) = \mcC(g_2) = \mcC(g)$, and the output gate has the same logical operation as $g$. Continuing this process, we obtain a primary boolean circuit consisting only of input and output gates. This maintains the property that if $g_1$, $g_2$ are inputs to an AND gate, then they are inputs to an OR gate as well. Thus, this primary boolean circuits must consist of components with four gates each: inputs $g_1$ and $g_2$, an output OR gate $g_3$, and an output AND gate $g_4$. We now construct our square certificate by including a square $(\{\mcC(g_1),\mcC(g_2)\},\mcC(g_3), \mcC(g_4))$ for each such component in the boolean circuit.  It is easy to see that for any input or output set $T_i \in \mcD$ in the constructed square certificate, $m(T_i) - tb(T_i) = n_i^{\ip} - n_i^{\op}$. Similarly, for any input or output gate $g$ in boolean certificate with $\mcC(g) = T_i \in \mcD$, we have $n_i^{\ip} - n_i^{\op} =  m(T_i) - tb(T_i)$. Thus $\sum_{i \in [n]} f_i \left( tb(T_i) - m(T_i) \right)  = \sum_{i \in [n]} f_i \left( n_i^{\op} - n_i^{\ip} \right) > 0$. Therefore, the multiset of squares thus obtained is indeed a square certificate. Further, if $S$ is an involved set, then $S = \mcC(g)$ for some gate $g$ in the boolean certificate, as required.
\end{proof}

Recall that in a cyclic boolean circuit, a gate may not get fixed to some value. %However, the boolean circuit of a boolean certificate has a special property -- if gates $g$ and  $g'$ are inputs to an AND gate, then these are inputs to an OR gate as well.
The next lemma crucially shows that any assignment to the input gates of a boolean circuit (of a boolean certificate) always fixes the values at the output gates, and hence for any satisfying assignments $A(\cdot)$, $A'(\cdot)$ to the gates that assign the same values to the input sets, the values assigned to the output gates must be the same as well. 

The proof requires the special structure of the boolean circuit in our boolean certificate, namely that if $g$, $g'$ are inputs to an AND gate, they are also inputs to an OR gate. First, since intermediate gates in the circuit have fan-in = fan-out, and fan-out for input gates = fan-in for output gates = 2, $|\ip| = |\op|$. We will show that  $\op \subseteq \fix^{(x_1,\dots,x_{|\ip|})}$ for any assignment  $(x_1,\dots,x_{|\ip|})$ to the input gates. That is, all the output gates get fixed for any assignment to the input gates.

We show this by giving an algorithm (Algorithm~\ref{fixing}) that takes as input an assignment $(x_1,\dots,x_{|\ip|})$ to the input gates, and assigns value $0$ or $1$ to some subset of gates $\all' \supseteq \op$ by setting $\val(g)$ for gates $g \in \all'$. We will prove that all gates in $\all'$ are in $\fix^{(x_1,\dots,x_{|\ip|})}$ and that $\fix^{(x_1,\dots,x_{|\ip|})}(g) = \val(g)$.

Initially, $\all' = \ip$. The algorithm maintains the invariant that for each gate $g \in \all'$, there is a value $\val(g)$ and a proof in the sense of Definition~\ref{def:fixing} that $\fix^{(x_1,\dots,x_{|\ip|})}(g) = \val(g)$ using gates that were assigned values before $g$. It starts with a gate $X_i \in \all'$ as the current gate. If the current gate is not an output gate, then it is an input to an OR gate (say $g_2$) and an AND gate (say $g_3$). Further, let $g_1$ be the other input to these gates (such a gate must exist, by property of boolean certificates). Then it uses properties from Definition~\ref{def:fixing} to show that at least one of $g_2$, $g_3$ has not been assigned a value, and can be chosen as the current gate and assigned a value so that the invariant is maintained. If the (new) current gate is not an output gate, the algorithm continues by considering the AND and OR gates that the current gate is input to.

\begin{algorithm}
	\caption{Fixing Algorithm($x_1,\dots,x_{|\ip|}$)}\label{fixing}
	\begin{algorithmic}[1]
		%\Procedure{Euclid}{$a,b$}\Comment{The g.c.d. of a and b}
		\State $\val(g) \gets \bot$ for all gates $g \in \all$
		\For{$i=1$ to $n$}%\Comment{sum(i)}
		\State $\val(X_i) \gets x_i$; $CG \gets X_i$;
		\While{$CG$ is not an output gate}%\Comment{We have the answer if r is 0}
		\State Let $g_2$, $g_3$ be the OR and AND gate that $CG$ is input to, and let $g_1$ be the other input to these gates 
		\If{$\val(g_2) \neq \bot$ and $\val(g_3) = \bot$}
		\State $\val(g_3) \gets \val(CG)$; $CG \gets g_3$;
		\ElsIf{$\val(g_2) = \bot$ and $\val(g_3) \neq \bot$} 
		\State $\val(g_2) \gets \val(CG)$; $CG \gets g_2$;
		\ElsIf{$\val(g_2) = \val(g_3) = \bot$ and $\val (CG) = 0$}
		\State $\val(g_3) \gets 0$; $CG \gets g_3$;
		\ElsIf{$\val(g_2) = \val(g_3) = \bot$ and $\val (CG) = 1$}
		\State $value(g_2) \gets 1$; $CG \gets g_2$;
		\EndIf
		\EndWhile\label{fixingendwhile}
		\EndFor
		
	\end{algorithmic}
\end{algorithm}

\begin{lemma}
	\label{fixed-output}
	Let $\all'$ be the set of all gates $g$ such that $\val(g)$ is set by the  Fixing Algorithm. Then  $\op \subseteq \all' $, and for all gates $g  \in \all'$, $g $ is in  $\fix^{(x_1,\dots,x_{|\ip|})}$ and $\fix^{(x_1,\dots,x_{|\ip|})}(g) = \val(g)$.
\end{lemma}

\begin{proof}
	The algorithm starts from gate $X_1$, and picks a sequence of gates and assigns them a value. Assume first that if we enter the while loop, at least one of the conditions from lines 6, 8, 10, or 12 must be true. Then starting with an input gate $X_i$ (i.e., in the $i$th iteration of the for loop), the algorithm assigns values to distinct gates that did not earlier have a value, until it assigns a value to an output gate. Then for each $X_i$, a distinct output gate is assigned a value, and since $|\ip| = |\op|$, $\op \subseteq \all'$.

	We now show one of the conditions from lines 6, 8, 10, or 12 must be true at any step, and that for all gates $g$ in $\all'$, we have a proof of $g$ being fixed to $\val(g)$.  We claim that the following statement is an invariant: if gate $g$ is assigned a value, then $g$ has a proof of being fixed to this value consisting only of gates that were assigned values before $g$. To show this, note that after the execution of $3$, the invariant will continue to hold if it was true previously, as input gates are trivially fixed to their value.  We show that if we enter the while loop with the  invariant condition then one of the conditions in lines 6, 8, 10, or 12 must be true and the invariant must hold true after execution of lines $5$ to $14$. This will complete the proof. Let $g$ be the current gate $CG$. Then $g$ must be assigned a value in the previous iteration of the while loop. Since we now enter the while loop, $g$ is not an output gate. Let $g_2$, $g_3$ be the OR and AND gates for which $g$ is an input, and let $g_1$ be the other input to these gates. Suppose for a contradiction that $\val(g_2) \neq \bot$ and $\val(g_3) \neq \bot$. By the invariant condition, both $g_2$ and $g_3$ have a proof of being fixed that does not use gate $g$. But $g_2$ is an OR gate, and can have a proof without $g$ only if $g_1$ is fixed to $1$, and $g_3$ can have proof without $g$ only if $g_1$ is fixed to $0$. By Lemma~\ref{fixed-value}, this is a contradiction. Thus one of the conditions 6, 8, 10, or 12 must hold true.
	
	We now show that the invariant condition holds after execution of one of the lines 6, 8, 10, 12. Suppose condition $6$ holds (the case when condition 8 holds is very similar). Then $g_3$ is assigned a value in the current while loop. Since $g_2$ is an OR gate and has a value assigned previously, by the invariant condition $g_1$ must previously (before $g_2$ was assigned its value) been assigned $\val(g_1)=1$. We then get the proof of $g_3$ being fixed to $\val(g_3)$ by making  the root of proofs of $g$ and $g_1$ the children of $g_3$ when $\val(g) = 1$, or the root of the proof of $g$ as the child of $g_3$ when $\val(g) = 0$. 
	
	Suppose condition $10$ holds (the case when condition 12 holds is very similar). Then gate $g_3$ is assigned a value in this while loop. Since $g_3$ is an AND gate, we get a proof of its being fixed to $\val(g_3) = 0$ by making the root of proof of $g$ the only child of $g_3$. 
\end{proof}

A function $f:\{0,1\}^n \rightarrow \{0,1\}$ is a monotone function if for any two assignments $(x_1,\dots,x_n)$ and $(x'_1,\dots,x'_n)$ such that for all $i \in [n], x_i \le x'_i$, $f(x_1,\dots,x_n) \le f(x'_1,\dots,x'_n)$. It is well known that a function $f$ is monotone iff $f(x_1,\dots,x_n) = \sum_{S \subseteq [n]} \alpha_S \prod_{j \in S} x_j$ for some $\alpha_S \in \{0,1\}$ for all $S \subseteq [n]$, and hence $f(\cdot)$ can be written as the OR of the AND of the subsets $S$ with $\alpha_S = 1$.

Given a boolean certificate, we associate with each gate $g$ a function $f_g:\{0,1\}^{|\ip|} \rightarrow \{0,1\}$ as follows. $f_g(x_1, \ldots, x_{|\ip|}) = b \in \{0,1\}$ if the assignment $(x_1,\dots,x_{|\ip|})$ to the input gates fixes gate $g$ to $b$, else $f_g(x_1, \ldots, x_n) = 1$. The next two lemmas show that for each gate $g$, the function $f_g$ is a monotone function, and that for any $(x_1,\dots,x_{|\ip|}) \in \{0,1\}^{|\ip|}$, the assignment $A(g) = f_g(x_1,\dots,x_{|\ip|})$ is a satisfying assignment. These allow us to complete the proof of Lemma~\ref{submodular-lemma}.

\begin{lemma}
	\label{monotone functions}
	The function $f_g$ is a monotone  for each intermediate and output gate $g \in \op \cup \im$.
\end{lemma}

\begin{lemma}
	\label{a satisfying-assignment}
	For any assignment $(x_1,\dots,x_{|\ip|})$ to the input gates, the extension $A$ given by $A(X_i) = x_i$ for all $X_i \in \ip, A(g) = f_g(x_1,\dots,x_{|\ip|})$ for all $g \in \op \cup \im$ % and  $A(z) = g_z(x_1,\dots,x_{n})$ for all $z \in Z$,
	is a satisfying assignment.
\end{lemma}

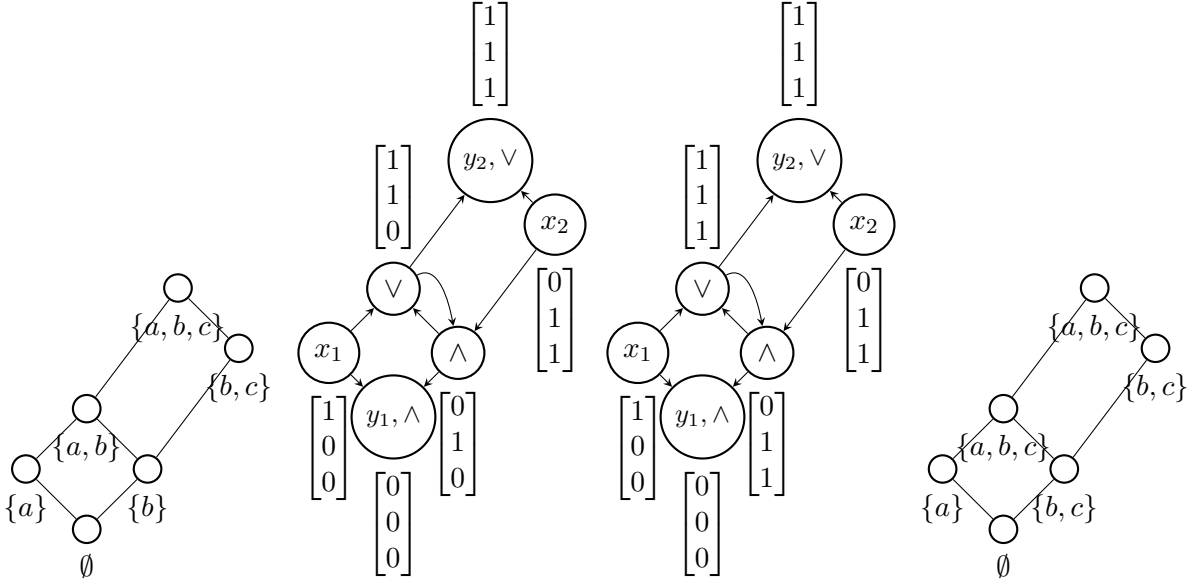
\begin{figure}[!ht]
	%\centering
	\begin{tikzpicture}[scale = .8]
	\node[style = {circle,thick,draw},label = below :$\emptyset$]  (a) at (0, 0) {};
	\node[style = {circle,thick,draw},label = below :$\{b\}$]  (b) at ( 1, 1) {};
	\node[style = {circle,thick,draw},label = below :$\{a\}$]  (c) at ( -1, 1) {};
	\node[style = {circle,thick,draw},label = below :${\{a,b\}}$]  (d) at ( 0, 2) {};
	\node[style = {circle,thick,draw},label = below :${\{b,c\}}$]  (e) at ( 2.5, 3) {};
	\node[style = {circle,thick,draw},label = below :${\{a,b,c\}}$]  (f) at ( 1.5, 4) {};
	%		\node[style = {circle,thick,draw},label = below :$\{a\}$]  (c) at ( -1, 1) {};
	%			\node[style = {circle,thick,draw},label = above :$\{a,b\}$]  (d) at (0, 2) {};
	% \node[style = {circle,thick,draw},label = below:$\{a,b\}$]  (a) at ( 0, 2) {};
	% \node[style = {circle,thick,draw}]  (e) at ( 2, 3) {$\{b,c\}$};
	%\node[style = {circle,thick,draw}]  (f) at ( 2.5, 4) {$\{a,b,c\}$};
	%\node[style = {circle,thick,draw}]  (a) at ( 1, 0) {$\{a\}$};
	%\node[style = {circle,thick,draw}]  (a) at ( 1, 0) {$\{a\}$};
	\draw (a) -- (b);
	\draw (a) -- (c);
	\draw (b) -- (d);
	\draw (c) -- (d);
	\draw (b) -- (e);
	\draw (d) -- (f);
	\draw (e) -- (f);		
	\end{tikzpicture}
	\begin{tikzpicture}[scale = .85]
	\node[style = {circle,thick,draw},label = below :{$\begin{bmatrix} 0 \\ 0 \\0 \end{bmatrix}$}]  (a) at (0, 0) {\small $y_1,\land$};
	\node[style = {circle,thick,draw},label = below :{$\begin{bmatrix} 0 \\ 1 \\0 \end{bmatrix}$}]  (b) at ( 1, 1) {$\land$};
	\node[style = {circle,thick,draw},label = below :{$\begin{bmatrix} 1 \\ 0 \\0 \end{bmatrix}$}]  (c) at ( -1, 1) {$x_1$};
	\node[style = {circle,thick,draw},label = above :{$\begin{bmatrix} 1 \\ 1 \\0 \end{bmatrix}$}]  (d) at ( 0, 2) {$\lor$};
	\node[style = {circle,thick,draw},label = below :{$\begin{bmatrix} 0 \\ 1 \\1 \end{bmatrix}$}]  (e) at ( 2.5, 3) {$x_2$};
	\node[style = {circle,thick,draw},label = above :{$\begin{bmatrix} 1 \\ 1 \\1 \end{bmatrix}$}]  (f) at ( 1.5, 4) {\small $y_2,\lor$};
	\draw[->,>=stealth] (b) -- (a);
	\draw[->,>=stealth] (c) -- (a);
	\draw[->,>=stealth] (b) -- (d);
	\draw[->,>=stealth](c) -- (d);
	\draw [->,>=stealth] (d) to [out=30,in=100] (b);
	\draw[->,>=stealth] (e) -- (b);
	\draw[->,>=stealth] (d) -- (f);
	\draw[->,>=stealth] (e) -- (f);		
	\end{tikzpicture}
	\begin{tikzpicture}[scale = .85]
	\node[style = {circle,thick,draw},label = below :{$\begin{bmatrix} 0 \\ 0 \\0 \end{bmatrix}$}]  (a) at (0, 0) {\small $y_1,\land$};
	\node[style = {circle,thick,draw},label = below :{$\begin{bmatrix} 0 \\ 1 \\1 \end{bmatrix}$}]  (b) at ( 1, 1) {$\land$};
	\node[style = {circle,thick,draw},label = below :{$\begin{bmatrix} 1 \\ 0 \\0 \end{bmatrix}$}]  (c) at ( -1, 1) {$x_1$};
	\node[style = {circle,thick,draw},label = above :{$\begin{bmatrix} 1 \\ 1 \\1 \end{bmatrix}$}]  (d) at ( 0, 2) {$\lor$};
	\node[style = {circle,thick,draw},label = below :{$\begin{bmatrix} 0 \\ 1 \\1 \end{bmatrix}$}]  (e) at ( 2.5, 3) {$x_2$};
	\node[style = {circle,thick,draw},label = above :{$\begin{bmatrix} 1 \\ 1 \\1 \end{bmatrix}$}]  (f) at ( 1.5, 4) {\small $y_2,\lor$};
	\draw[->,>=stealth] (b) -- (a);
	\draw[->,>=stealth] (c) -- (a);
	\draw[->,>=stealth] (b) -- (d);
	\draw[->,>=stealth](c) -- (d);
	\draw [->,>=stealth] (d) to [out=30,in=100] (b);
	\draw[->,>=stealth] (e) -- (b);
	\draw[->,>=stealth] (d) -- (f);
	\draw[->,>=stealth] (e) -- (f);		
	\end{tikzpicture}
	\begin{tikzpicture}[scale = .8]
	\node[style = {circle,thick,draw},label = below :$\emptyset$]  (a) at (0, 0) {};
	\node[style = {circle,thick,draw},label = below :${\{b,c\}}$]  (b) at ( 1, 1) {};
	\node[style = {circle,thick,draw},label = below :$\{a\}$]  (c) at ( -1, 1) {};
	\node[style = {circle,thick,draw},label = below :${\{a,b,c\}}$]  (d) at ( 0, 2) {};
	\node[style = {circle,thick,draw},label = below :${\{b,c\}}$]  (e) at ( 2.5, 3) {};
	\node[style = {circle,thick,draw},label = below :${\{a,b,c\}}$]  (f) at ( 1.5, 4) {};
	%		\node[style = {circle,thick,draw},label = below :$\{a\}$]  (c) at ( -1, 1) {};
	%			\node[style = {circle,thick,draw},label = above :$\{a,b\}$]  (d) at (0, 2) {};
	% \node[style = {circle,thick,draw},label = below:$\{a,b\}$]  (a) at ( 0, 2) {};
	% \node[style = {circle,thick,draw}]  (e) at ( 2, 3) {$\{b,c\}$};
	%\node[style = {circle,thick,draw}]  (f) at ( 2.5, 4) {$\{a,b,c\}$};
	%\node[style = {circle,thick,draw}]  (a) at ( 1, 0) {$\{a\}$};
	%\node[style = {circle,thick,draw}]  (a) at ( 1, 0) {$\{a\}$};
	\draw (a) -- (b);
	\draw (a) -- (c);
	\draw (b) -- (d);
	\draw (c) -- (d);
	\draw (b) -- (e);
	\draw (d) -- (f);
	\draw (e) -- (f);
	\end{tikzpicture}
	\caption{(i) A square certificate where involved sets $\{a,b\}$ and $\{b\}$ are not in the lattice closure of defined sets $\{\emptyset,\{a\},\{b,c\},\{a,b,c\}\}$. (ii) Corresponding boolean certificate, the labels are characteristic vectors of the creator function (and satisfying assignments). (iii) Boolean certificate with modified satisfying assignments. (iv) New square certificate.}
	\label{fig:conversion}
\end{figure}

\begin{proof}[Proof of Lemma~\ref{submodular-lemma}.] We prove the second property, since the first was proved earlier. If there is a square certificate for a partial function, then there exists a boolean certificate. From the definition of boolean certificates, for each $i \in [m]$, the assignment to all gates in which gate $g \in \all$ is assigned $ (\mcC(g))_i$  is a satisfying assignment. From Lemmas \ref{intutive-defn} and \ref{fixed-output}, we have $\fix^{((\mcC(X_1))_i,\dots,(\mcC(X_{|IP|})_i)}(g) = (\mcC(g))_i$ for all output gates $g$ and $i \in [m]$ (Recall that $\ip = \{X_1,\dots,X_{|IP|}\}$). Consider the assignment $B^i: \all \rightarrow \{0,1\}$, for each $i \in [m]$, given by $B^i(g) = (\mcC(g))_i$ for all $g \in \ip$ and $B^i(g) = f_g(((\mcC(X_1))_i,\dots,(\mcC(X_{|IP|})_i))$ for all $g \in \op \cup \im$. Hence, $B^i(g) = (\mcC(g))_i$ for all $g \in \ip \cup \op$ and $i \in [m]$.  Consider the function $\mcC':\all \rightarrow 2^{[m]}$ given by $\mcC'(g) = \{i \in [m]| B^i(g) = 1\}$. Note that $\mcC'(g) = \mcC(g)$ for all $g \in \ip \cup \op$ and hence $\mcC'$ satisfies the second property of boolean certificate. Also $(\mcC'(g))_i$ is a satisfying assignment for each $i \in [m]$ by Lemma \ref{a satisfying-assignment}. Therefore, there exists a boolean certificate with creator function $\mcC'$ instead of $\mcC$. By Lemma~\ref{lem:squareboolean}, there exists a square certificate with set of involved sets as $\{\mcC'(g) | g \in \all\}$. Since $f_g$ is a monotone function of the values at the input gates,  we have for all $i \in [m]$,  $B^i(g) = \sum_{S \subseteq \ip} \alpha_S \prod_{j \in S} (\mcC(X_j))_i$ for some $\alpha_S \in \{0,1\}$. Therefore, for all $g \in \all$, $\mcC'(g)$ can be obtained by union and intersection of sets $\mcC(X_j)$ for $X_j \in \ip$, i.e., of input sets.   
\end{proof}
		\section{Convex functions}
A function $f: K \rightarrow \mathbb{R}$ is convex ($K \subseteq \mathbb{R}^m$) iff for all $x,y \in K$ and $\lambda \in [0,1]$,
\[f(\lambda x + (1 - \lambda)y) \le \lambda f(x) + (1 - \lambda)f(y).\]
The problem of extending a convex function is extensively studied in convex analysis. We briefly discuss the work of Yan~\cite{Yan12} and Dragomirescu and Ivan~\cite{DragomirescuI92}, focusing on the presentation by Yan. Given a partial function $f$ that is convex on a non-convex domain $C$, Yan considers extending $f$ to a convex function both within and outside the convex hull $\conv(C)$. For a point $x \in \conv(C)$, he defines function $\hg$ using the well-known convex roof construction:
\begin{align}
\hg(x) = \inf \left\{ \sum_{y \in C} \lambda_y f(y) \, : \, \lambda \ge 0, \, \sum_{y \in C} \lambda_y =1, \, \mbox{ and } \sum_{y \in C} \lambda_y y = x \right\} \, . 
\label{eqn:hgx}
\end{align}

\noindent The infimum thus runs over all possible convex combinations of points $y \in C$ that evaluate to $x$. If the function $f$ is bounded below or the domain $C$ contains a point in the relative interior of $\conv(C)$,  then $\hg$ is a convex extension~\cite{Yan12}.  For a point $x$ outside the convex hull of $C$, assuming now that $f$ is defined  and convex inside  $\conv(C)$, Yan defines
\begin{align}
\tg(x) = \sup \left\{ \lambda f(y) + (1-\lambda) f(z) : x = \lambda y + (1-\lambda) z, ~ y,z \in \conv(C), \lambda \ge 1 \right\} \,. 
\label{eqn:tgx}
\end{align}

The function $\tg$ is a convex extension iff $f$ satisfies the Lipschitz property\footnote{A function $f$ has the Lipschitz property on $K$ if there exists a constant $L$ such that $|f(x) - f(y)| \le L ||x-y|| \quad \forall x,y \in K$} on $\conv(C)$~\cite{Yan12}. The extensions $\hg(x)$ and $\tg(x)$ are optimal in the following sense. Within the convex hull, $\hg(x)$ is maximal, i.e., for any convex function $g$ that extends $f$ to the convex hull of $C$, for any point $x$ inside the convex hull, $g(x) \le \hg(x)$. This is because for optimal $\lambda$ in the definition of $\hg$, $g(x) \le \sum_{y \in C} \lambda_y  g(y) = \sum_{y \in C} \lambda_y  f(y) = \hg(x)$.  Similarly, outside the convex hull, $\tg(x)$ is  minimal, i.e., for any convex function $g$ that extends $f$ (assuming $f$ is defined and convex on  the convex hull) to $\mathbb{R}^m$,  $g(x) \ge \tg(x)$ for any point $x$ outside the convex hull. This again can be easily seen as for optimal $\lambda$ and $y,z\in \conv(C)$ in the definition of $\tg$, $g(x) \ge \lambda g(y) + (1-\lambda) g(z) = \lambda f(y) + (1-\lambda) f(z) = \tg(x)$.

\paragraph*{Our results.} The set of points $\mathcal{D} = \{T_1, \dots, T_n\}$ in the given partial function $H$ corresponds to the domain $C$ described above. We assume for simplicity that $\conv(\mcD)$ has non-zero volume. We show  that Approximate Convex Extension (and hence Convex Extension) can be solved in polynomial time. We also give a unified construction for a convex function $\tf$ that equals $\hg(x)$ inside the convex hull $\conv(\mcD)$ and $\tg(x)$ outside, and show that if there exists a convex extension of $H$, then our construction is a convex extension. However, we show that evaluating $\tf$ for a point $x \in \conv(\mcD)$ is strongly NP-hard. Our results hold for concave functions as well, using the fact that $f$ is a convex function iff $-f$ is concave. Recall that $H = \{(T_1,f_1),\dots,(T_n,f_n)\}$ is our partial function, with each $T_i \in \mathbb{R}^m$ and $f_i \in \mathbb{R}.$

We now give the construction of our convex extension $\tf$. Since our partial function is defined on a finite set of points $\mcD$, the convex roof function $\hg(x)$ from~\eqref{eqn:hgx} is  the optimal value of the linear program Convex-P:

\[
\mbox{Convex-P: } \min \displaystyle\sum_{i = 1}^{n} \lambda_i f_i, \mbox{ s.t. } \displaystyle\sum_{i = 1 }^{n} \lambda_i T_i = x, \quad \displaystyle\sum_{i = 1 }^{n} \lambda_i  =  1, \mbox{ and }  \lambda_i \ge 0 \quad \forall i \in [n] \, .
\]
\[
\mbox{Convex-D: } \max \mu + \displaystyle\sum_{j = 1}^{m} x_j y_j \mbox{ s.t. } \displaystyle\sum_{j = 1 }^{m} (T_i)_j y_j  + \mu \le f_i \quad \forall i \in [n] \, .
\]

\noindent The linear program Convex-D is dual of  Convex-P and has variables $(y_1, \dots, y_m,\mu)$.
% Note that Convex-P is feasible iff $x \in \conv(\mcD)$, and can be solved in polynomial time.
  Now let $Q = \{(y_1, \dots, y_m,\mu)|\sum_{j = 1 }^{m} (T_i)_j y_j  + \mu \le f_i \quad \forall i \in [n]\}$ denote the polyhedron of dual feasible solutions. By our assumption that $\conv(\mcD)$ has non-zero volume, $Q$ has at least one vertex (shown in Appendix \ref{sec:convexappendix}). Let $V = \{v_1,\dots,v_N\}$ be the set of vertices of the polyhedron $Q$, and let $v_i = (y^i,\mu^i)$ for all $i \in [N]$, with $y^i \in \mathbb{R}^m$ and $\mu^i \in \mathbb{R}$. We define $\tf$ as the maximum over all vertices of $Q$, of the objective of Convex-D:

\begin{align}
%\label{convex-inequality}
	\tf(x) = \max_{1 \le i \le N}  \left \{\displaystyle\sum_{j = 1}^{m} x_j y^i_j  + \mu^i \right \} \, .
	\label{def:tf}
\end{align}

\noindent Since $\tf(x)$ is the maximum of linear functions, it is convex. Also, $\tf(x) = \hg(x)$ if $x \in \conv(\mcD)$. This is because Convex-P is feasible for $x \in \conv(\mcD)$, and hence the dual is bounded and  there is an optimal  vertex.
\begin{lemma}
	\label{convexext-iff}
	A partial function $H$ can be extended to a convex function on $\mathbb{R}^m$ iff  $ \tf(T_i) = f_i$ for all $i\in [n]$.
\end{lemma}	
\begin{proof}
	Since $\tf(x)$ is convex so if $\tf(T_i) = f_i$ for all $i \in [n]$ then clearly $\tf$ is a convex extension to $\mathbb{R}^m$.	
	
	Now suppose that $\tf(T_k) = \hg(T_k)$ is not equal to $f_k$ for some $k \in [n]$.
	For any $x = T_i$, it is clear that $\hg(x) \le f_i$ (if we set  $\lambda_i = 1$, then the objective value is $f_i$). Thus $\hg(T_k) < f_k$.  If there exists a convex extension $g$ then $g(T_k)  \le \hg(T_k) < f_i$, which is a contradiction.  The first inequality is due to the earlier observation that if $g$ is any convex extension of $H$ on $\conv(\mcD)$ then $g(x) \le \hat{g}(x)$ for all $x \in \conv(\mcD)$. 
\end{proof}	
\begin{theorem}
	\label{convexextinP}
	Convex Extension is in P, and if $H$ is extendible then $\tf$ is an extension.
\end{theorem}	
\begin{proof}
	By Lemma \ref{convexext-iff}, determining the existence of extension boils down to checking  $ \tf(T_i) = \hg(T_i) = f_i$ for all $i\in [n]$. Since $\hg$ can be efficiently computed by solving Convex-P, Convex Extension is in P. If $H$ is extendible then we have by Lemma \ref{convexext-iff}, $ \tf(T_i) = f_i$ for all $i\in [n]$. Hence, $\tf$ is an extension. 
\end{proof}

\begin{theorem}
	\label{convexextension}
	Approximate Convex Extension is in P.%, and the function $\tilde{f}$ defined by (\ref{def:tf}) is a convex extension of $H$ when $H$ is extendible.
\end{theorem}
We have shown that $\tf$ is the maximal extension  $\hg$  inside the $\conv(\mcD)$. Interestingly, $\tf(x)$  is also the minimal extension $\tg(x)$  outside the convex hull, where $\tg(x)$ is defined as follows. 
\begin{align*}
\tg(x) =  \sup \left\{ \lambda \hg(y) + (1-\lambda) \hg(z) : x = \lambda y + (1-\lambda) z, y,z \in \conv(\mathcal{D}), \lambda \ge 1 \right\} \, .
\end{align*}
\begin{lemma} \label{minimal-extension}
	For any partial function $H = \{(T_1,f_1),\dots,(T_n,f_n)\}$, we have $\tf(x) = \tg(x)$ for all $x$ outside the convex hull $\conv(\mcD)$. 
\end{lemma}

The function $\tf(x)$ is thus a natural and canonical convex extension of a given partial function. In fact, we know of no other convex extensions that are widely studied. It is then natural to ask, given a partial function, to evaluate the convex extension $\tf(x)$ at a given point. Surprisingly, we show that this problem is strongly \nphard, and hence this canonical extension cannot be efficiently evaluated at a given point, unless P = NP. Since $\tf$ is the maximum over vertices $V = \{v_1,\dots,v_N\}$ of the polyhedron $Q$, one may wonder why it can not be computed by solving the linear program Convex-D.  This is because, for $x$ outside $\conv(\mcD)$, Convex-P is infeasible and  hence Convex-D is unbounded. So the optimal value is $\infty$, whereas $\tf(x)$ is a finite.

 We show a reduction from the \textsf{Optimal Vertex} problem, which is  strongly NP-hard~\cite{optimal-vertex}. In this problem, we are given an $n \times m$ rational matrix $A$, rational $n$-vector $b$, $m$-vector $c$ and a rational number $K$. Then the objective is to  decide if there exists a vertex $v$ of the polyhedron $A y \le b$ with $c^T v \ge K$.

In the proof of hardness of \textsf{Optimal Vertex} in \cite{optimal-vertex}, the instance $A$ and $b$ satisfy the property that the polyhedron $A y \le b$ has at least one vertex. We will use this property in our proof. Before we give the reduction, we state a  property of vertices of polyhedra.
\begin{lemma}[Lemma 8.2 of \cite{bertsimas}\footnote{The first part of this lemma is 8.2(a)  of \cite{bertsimas} whereas the second part   follows from the proof of 8.2(a)}]  
	\label{polytope-property}
	Suppose $A$ is a $n \times m$  matrix and  $b$ is an $n$-vector. Also assume all the entries of $A$ and $b$ are integers and $U$ is the largest absolute value of the entries in $A$ and $b$. Then every extreme point of the polyhedron $P = \{y \in \mathbb{R}^m | Ay \le b\} $ satisfies : 
	\begin{enumerate}
		\item {$|y_j| \le (mU)^m, j = 1,\dots,m$ }
		\item{If $y_j \neq 0$ then $|y_j| \ge \frac{1}{(mU)^m}, j = 1,\dots,m$}
	\end{enumerate}
\end{lemma}
\begin{theorem}
	\label{convexcomputation}
	Given a partial function $H$, a point $x \in \mathbb{R}^m$ and rational $k$, it is strongly NP-hard  to determine if $\tf(x) \ge k$. However, if the dimension $m$ is constant, $\tilde{f}(x)$ can be computed in polynomial time.
\end{theorem}
\begin{proof}
		First note that the number of vertices of the dual polyhedron $Q$ is bounded by $O(n^m)$, and can be enumerated~\cite{vertices}. If $m$ is constant, this gives a polynomial time algorithm to compute $\tf(x)$.
		
	Let the  instance of \textsf{Optimal Vertex} have parameters   $A$, $b$, $c$ and  $K$, and as noted  we  assume that the polyhedron $A y \le b$ has at least one vertex. We will also assume wlog that all entries in $A$ and $b$ are integers. Also let $U$ be the largest absolute value among all entries of $A$ and $b$. Let $A_i = [a_{i1} a_{i2} \dots a_{im}]^T$ be the $i$th row of matrix $A$.
	
	Let $M = (mU)^m$. We construct the instance for our problem as follows. We set the partial function $H = \{(A_1,b_1),\dots,(A_n,b_n),(\bm{0},0)\}$,  $x =\frac{c}{L}$ and $k = K/L$, where $L = 3 M^2 \sum_{j = 1}^{m} |c_j|$ (we need $L > 2 M^2 \sum_{j = 1}^{m} |c_j|$). Note that the size of $x$ is polynomial in the sizes of $A,b$ and $c$.  Our claim is that $\tilde{f}(x) = \frac{\max_{v} c^T v}{L}$ where maximum is taken over all vertices of polyhedron $A y \le b$. Thus  $\tilde{f}(x) \ge k $ iff there exists a vertex $v$ such that $c^T v \ge K$. To prove the claim, we observe that the polyhedron $Q$ associated with the definition of $\tilde{f}$ is $\{A y  +\bm{\mu} \le b, \mu \le 0\}$  (where $\bm{\mu}$ is an $n$-vector with each entry $\mu$).  Let $Q'$ be the polyhedron $A y \le b$ and $V'$ be the set all vertices of $Q'$ ($V'$ is non-empty). It is easy to see that $\hat{y}$ is a vertex of $Q'$ iff $(\hat{y},0)$ is a vertex of $Q$. Therefore, the polyhedron $Q$ also has at least one vertex. Let $V$  be the set of all vertices of $Q$. We denote any vertex in $Q$ as $(y,\mu)$.  Since $V$ is non-empty,  we have $\tilde{f}(x) = \tilde{f}(\frac{c}{L}) = \max_{(y,\mu) \in V}  (\frac{c}{L})^T y + \mu$ from ~\eqref{def:tf}. If we prove that this maximum can only be attained at the vertices with $\mu = 0$ then we will be done. Note that any vertex $(y,0) \in V$ has the property $  (\frac{c}{L})^T y + \mu \ge -\frac{M \sum_{j} |c_j|}{L}$ (because of first part of Lemma \ref{polytope-property}). Consider a vertex $(y,\mu) \in V$ with $\mu < 0$ (recall that $\mu \le 0$). For such a vertex $   (\frac{c}{L})^T y + \mu \le   \frac{M \sum_{j}|c_j|}{L} -\frac{1}{M}$ (Lemma \ref{polytope-property}). For our choice of $L$, $ -\frac{M \sum_{j} |c_j|}{L} >  \frac{M \sum_{j}|c_j|}{L} -\frac{1}{M} $. Hence,  $\tilde{f}(x) = \max_{v \in V'} \frac{c^T v}{L}$, completing the proof of the claim. 
	%Let $D_1$ be the maximization problem $\{\max_{} \mu + (\frac{c}{L})^T y \}
\end{proof}

	\vspace{0.1in}	
	\noindent \textbf{Conclusion.} Our work is the first to formally study the complexity of partial function extension. We show that results can often be counterintuitive, and shed new light on problems previously studied. Our work also gives a number of new results for learning and property testing.  While there are clearly a large number of interesting open problems, one we particularly would like to highlight is the basic question of membership in NP (or coNP) of partial function extension. We are able to resolve this for XOS, subadditive, and convex functions, but leave it open for submodular functions. Resolving this problem may lead to further insights on the structure of these functions.

	\bibliographystyle{plain}
	\bibliography{bib_agt}

	\appendix
		\section{Subadditive and XOS functions}
\label{sec:subadditiveappendix}

\subsection*{Proof of Theorem \ref{gen-ext-xos} }
Let the given partial function be $H = \{(T_1,f_1),\dots,(T_n,f_n)\}$.
We claim that the optimal value of $\alpha$ for Aproximate XOS Extension (say $\hat{\alpha}$) is equal to the optimal value of $\alpha$ in  the following linear program (say $\alpha^*$), with variables $\alpha$ and  $w_{ij}$ for all $1 \le i \le n$ and $1\le j \le m$. Since the linear program can be  solved in polynomial time, this claim implies that Approximate XOS Extension can be efficiently solved.
\[ \min \alpha\] 
\[f_i \le w_{i}^T \chi(T_i) \le \alpha f_i \quad \forall i \in [n]\]
\[w_{i}^T \chi(T_i) \ge w_{j}^T \chi(T_i) \quad \forall i,j \in [n]\]
\[w_i \in \mathbb{R}^m_+\]
\[\alpha \ge 1 \]
% Let $\epsilon^*$ be the optimal value of above linear program.
Let the XOS function $g$ corresponding to the optimal solution $\hat{\alpha}$ for Approximate XOS Extension be given by linear functions $v_1,\dots,v_k \in \mathbb{R}^m_+$ for some $k \ge 1$, and let the linear functions be indexed so that  $g(T_i) = v_{i}^T \chi(T_i)$ for  $i \in [n]$ (the same linear function can appear with  multiple indices, i.e., $v_i = v_j$ for $i \neq j$). Then $f_i \le g(T_i) =  v_{i}^T \chi(T_i) \le \hat{\alpha}f_i$ for all $ i \in [n]$, and $v_{i}^T \chi(T_i) \ge v_{j}^T \chi(T_i)$ for all $i,j \in [n]$. It is clear that $\hat{\alpha},\langle v_i \rangle_{i \in [n]}$ are feasible for the linear program, hence $\alpha^* \le \hat{\alpha}$.  By definition, $\hat{\alpha} \le \alpha^*$, since the linear program produces an XOS function that has value within $\alpha^*$ factor at each $T_i$ for all $i\in [n]$. Hence $\hat{\alpha} = \alpha^*$.
	
 \subsection*{Upper bound for Approximate Extension}
We now show the upper bound for the Approximate Extension. 	Let $\mathcal{F}$ and $\mathcal{G}$ be two classes of functions. We say that   $\mathcal{G}$ $\theta$-approximates $\mathcal{F}$ if for all functions $f$ in $\mathcal{F}$, there exists a function $g$ in $\mathcal{G}$ such that $g(S) \le f(S) \le \theta g(S) $ for all $S \subseteq [m]$. We first prove the following lemma.
 \begin{lemma}
 	\label{class-approx}
 	Let $\mathcal{F}$ and $\mathcal{G}$ be two classes of functions so that $\mathcal{G}$ $\theta_1$-approximates $\mathcal{F}$  and  $\mathcal{F}$ $\theta_2$-approximates $\mathcal{G}$. If there is a  $\rho$-approximation algorithm for  Approximate Extension  for $\mathcal{F}$ then there is an $\rho \theta_1 \theta_2$- approximation algorithm for  Approximate Extension  for $\mathcal{G}$.  
 \end{lemma}
 
\begin{proof}
For a given instance of partial function extension,  let $\alpha^*_\mathcal{F}$ and $\alpha^*_\mathcal{G}$ be the optimal value of $\alpha$ in the Approximate Extension problem for $\mathcal{F}$ and $\mathcal{G}$ respectively.  Let $A$ be the $\rho$-approximation algorithm  for $\mathcal{F}$. A $\rho \theta_1 \theta_2$-approximation algorithm for Approximate Extension  for $\mathcal{G}$ is as follows: given any partial function $H$, return $\theta_1 \alpha$ where $\alpha$ is the value returned by algorithm $A$ on $H$. We have $\alpha^*_\mathcal{F} \ge \frac{\alpha}{\rho}$ as $A$ is an $\rho$-approximation algorithm. Since $\mathcal{G}$ $\theta_1$-approximates  $\mathcal{F}$, we have  $\alpha^*_{\mathcal{G}} \le \theta_1 \alpha^*_{\mathcal{F}}$. As  $\alpha^*_{\mathcal{F}} \le  \alpha$ so we have $\alpha^*_{\mathcal{G}} \le \theta_1 \alpha$. Also $\mathcal{F}$ $\theta_2$-approximates  $\mathcal{G}$ so we have $\alpha^*_{\mathcal{F}} \le \theta_2 \alpha^*_{\mathcal{G}}$. Then $\alpha^*_{\mathcal{G}} \ge \frac{\alpha^*_{\mathcal{F}}}{\theta_2} \ge \frac{\alpha}{\rho \theta_2}$. Hence $\frac{\alpha}{\rho \theta_2} \le \alpha^*_{\mathcal{G}} \le \theta_1 \alpha$. This proves our result.	 
 \end{proof}
 
 Recall that  XOS functions are a subclass of subadditive functions. We will use the following result:

 \begin{theorem}[\cite{BhawalkarR11,Dobzinski07}]
 	For any subadditive function $f$, there exists an XOS function $g$ such $g(S) \le f(S) \le O(\log m)g(S)$ for all $S \subseteq [m]$.
 \end{theorem}
 
 From the above results and Lemma \ref{class-approx}, the upper bound for Theorem \ref{gen-ext-subadditive-approx} follows, with $\theta_1 = 1, \theta_2 = O(\log m)$ and $\rho = 1$.
 \subsection*{Proof of Lemma \ref{learning-application}}
 	Consider the distribution $\mu$ that assigns probability mass uniformly to $\mathcal{D} = \{T_1,\dots,T_n\}$ and $0$ elsewhere. We restrict the target function to the family $\mathcal{F'}  = \{f \in \mathcal{F} | f(T_i) \in [1,r] \thinspace \forall i \in [n]\} \subseteq \mathcal{F}$. Suppose there is an  algorithm that PMAC-learns $\mathcal{F'}$ with approximation factor $< r$. Let $g$ be an arbitrary function  in $\mathcal{F'}$. Let the input to the algorithm be $\{(S_i,g(S_i))\}_{1 \le i \le l}$  where  $l$ is $poly(m)$ (let $\epsilon$ and $\delta$ be constant). Let $\mathcal{F}^*  = \{h \in \mathcal{F'} | h(S_i) = g(S_i) \thinspace \forall i \in [l]\}$. By our assumption, given any values in $[1,r]$ at the sets  $\mathcal{D}\setminus \{S_1,\dots,S_l\}$, there is a function in $\mathcal{F}^*$ that takes those values. Let  $S \sim \mu$ and  $v$ be the value returned by the algorithm at $S$. If $\alpha < r$ then  $f^*(S) \le v \le \alpha f^*(S)$ for all target function $f^* \in \mathcal{F}^*$ iff $S \in \{S_1,\dots,S_l\}$. Since $l$ is polynomial (whereas $n$ is superpolynomial), the algorithm then returns a value within $\alpha$ factor with only small probability. Hence  $\alpha$  must be at least $r$. Also above argument  and hence lower bound holds
even if the algorithm knows the distribution $\mu$, allowed unbounded computation and  choose samples $(\{S_i\}_{1 \le i \le l})$  adaptively. 
	
\subsection*{Proof of Lemma \ref{lemma-general-subadditive}}	
	One direction is trivial. If there exists $T_1,\dots,T_r,\cup_{i = 1}^{r} T_i \in \mathcal{D}$ such that  $\sum_{i = 1}^{r} f(T_i) < f(\cup_{i = 1}^{r} T_i)$ then partial function is not extendible. Now assume this is not the case. Let $\mathcal{D}^c := \{S| S = \cup_{i = 1}^{r} A_i \quad \text{for some} \quad A_1,\dots,A_r \in \mathcal{D}\}$ be the union-closure of $\mathcal{D}$. We now define  $\hf$ which is an extension of $f$ to $\mathcal{D}^c$. If $S \in \mathcal{D}$ then  $\hf(S) = f(S)$. If $S \not \in \mathcal{D}$ (and $S \in \mathcal{D}^c$) then $\hf(S)$ is the minimum value of $\sum_{i =1}^{k} f(S_i)$ over all families of sets $(S_1,\dots,S_k)$ such that  each $S_i \in \mathcal{D}$, and $\cup_{i =1}^{k} S_i = S$.  Let $M$ be the maximum value of $\hf$ on $\mathcal{D}^c$. We define an extension of $\hf$ to $2^{[m]}$ by assigning value $M$ to each set not in $\mathcal{D}^c$. Let this extension be $\tf$. We claim that $\tf$ is subadditive.
	
	Note that $M$ is the maximum value of $\tf$. Let $A$ and $B$ be any two sets. If any of $A$ or $B$ is not in $\mathcal{D}^c$ then $\tf(A)+\tf(B)$ is at least $M$ and thus $\tf(A) + \tf(B) \ge \tf(A \cup B)$. Therefore, we assume both $A$ and $B$ are in $\mathcal{D}^c$ which implies $A \cup B$ is also in $\mathcal{D}^c$. Let $A$ be the union of $A_1,\dots,A_r \in \mathcal{D}$ ($r \ge 1$) and $B$ be the union of $B_1,\dots,B_{r'} \in \mathcal{D}$ ($r' \ge 1$). Therefore, $A \cup B$ is union of $A_1,\dots,A_r,B_1,\dots,B_{r'}$. If $A \cup B$ is in $\mathcal{D}$ then by assumption and otherwise by definition of $\hf$, we have $\tf(A) + \tf(B) \ge \tf(A \cup B)$. 
		\section{Submodular functions}
\label{sec:submodularappendix}

\subsection*{Proof of Lemma~\ref{vs-in-sc}.}
Assume for a contradiction that there exists a square certificate and an extension $f(\cdot)$ without a square tuple as described in the lemma. Summing the inequalities $f(A)  + f(B) \ge  f(A \cup B) + f(A \cap B)$ for each square tuple $(A,B,A\cup B, A \cap B)$ in the square certificate, we observe that all the intermediate sets cancel out, since they appear an equal number of times on the left and right hand sides. We get $\sum_{i \in [n]} f_i \thinspace m(T_i) \ge  \sum_{i \in [n]} f_i \thinspace tb(T_i)$ which contradicts property (P2) of the square certificate.

 \subsection*{Proof of Lemma~\ref{fixed-value}.}
  If gate $g$ is fixed to both $b'$ and $b''$, there must be two proofs $G'=(V',E')$ with values $\val'(\cdot)$, and $G''=(V'', E'')$ with values $\val''(\cdot)$ as described above. In the rooted tree $G'$, let $g_0$ be a gate at maximum distance from the root $g$ that has different values in the two proofs (such a gate must exist, since $g$ is such a gate). Gate $g_0$ cannot be a leaf, since all leaves have the same values by definition. Assume that $g_0$ is an AND gate; a similar proof holds if $g_0$ is an OR gate. If $g_0$ has one child in both proofs, since $g_0$ is an AND gate, by definition it must have value $0$ in both proofs. Similarly if $g_0$ has two children in both proofs, it must have value $1$ in both proofs. Now suppose without loss of generality that $g_0$ has one child (say $h_0$) in $G'$, and two children in $G''$. Then $\val'(g_0) = \val'(h_0) = 0$. Since the fan-in for each gate is 2 in the boolean circuit, $h_0$ must be a child of $g_0$ in $G''$ as well, and by assumption, $\val'(h_0) = \val''(h_0) = 0$. But this gives us a contradiction, since if $g_0$ is an AND gate with two children in any proof, then both children must have value 1.

  \subsection*{Proof of Lemma~\ref{intutive-defn}.}
Let $G'=(V',E')$ and $\val'(\cdot)$ be a proof for gate $g$ for the given assignment to the input gates. Let $g_0$ be the gate at maximum distance from the root for which $\val'(g_0) \neq A(g_0)$. The children of $g_0$ by assumption have the same value as in the satisfying assignment $A(\cdot)$, and considering the cases in the construction of the proof gives us a contradiction.

  \subsection*{Proof of Lemma~\ref{monotone}.}
  Let $G'=(V',E')$ and $\val'(\cdot)$ be a proof that fixes $\val(g') = 0$ for the assignment $(x_1, \ldots, x_{|\ip|})$. Then as observed previously, every gate in the proof has value 0, including the input gates as the leaves. Since $(x_1, \ldots, x_{|\ip|})  \ge (x_1', \ldots, x_{|\ip|}')$, the values at the leaves of the rooted tree remain unchanged, and the proof $G'=(V',E')$ and $\val'(\cdot)$ is a proof with inputs $(x_1', \ldots, x_{|\ip|}')$ as well, that fixes gate $g$ to $0$. The proof of the converse holds similarly.

  \subsection*{Proof of Lemma~\ref{monotone functions}.}
  	Consider an assignment in which an input $x_i$ is set to $0$ and $f_g$ is $1$ for some $g \in \op \cup \im$. If the gate $g$ is fixed to $1$ by the assignment then it will continue to be fixed to $1$ after setting $x_i$ to 1 by Lemma \ref{monotone}. Suppose instead that the gate $g$ is not fixed by the initial assignment with $x_i = 0$. In this case also, after setting $x_i =1$, gate $g$ cannot get fixed to $0$ by Lemma~\ref{monotone}. So, flipping any input $x_i$ from $0$ to $1$ can not make the function $f_g$ evaluate from $1$ to $0$. Hence, $f_g$ is monotone.
  	\subsection*{Proof of Lemma~\ref{a satisfying-assignment}.}

 	We will show that the boolean operation at any gate $g$ is not violated by this assignment. Let $g$ be an AND gate with incoming edges from gates $g_1$ and gate $g_2$ (the case when $g$ is an OR gate is similar). Suppose first that $g$ is fixed by the assignment $(A(X_1),\dots,A(X_n))$ to the input gates. Since $g$ is fixed, by Definition~\ref{def:fixing}, either all of $g$, $g_1$ and $g_2$ are fixed to $1$ or $g$ and one of $g_1$ and $g_2$  is fixed to $0$. In either case, the boolean operation is satisfied. Now assume that $g$ is not fixed and hence it is assigned $1$. For a violation, either $g_1$ or $g_2$ must be assigned $0$. Say $g_1$ is assigned 0. But then $g_1$ is fixed, and since $g$ is an AND gate, $g$ must also be fixed to 0, giving a contradiction.

		\section{Convex functions}
\label{sec:convexappendix}

\subsection*{Proof that $Q$ has at least one vertex }

\begin{claim}
 If $\conv(\mcD)$ has non-zero volume then the polyhedron $Q$ has at least one vertex.
\end{claim}
\begin{proof}
	 Since $\conv(\mcD)$ has non-zero volume, the  $T_i$'s do not lie on a hyperplane. Formally, this means that  $(y,\mu)( y \in \mathbb{R}^{m}, \mu \in \mathbb{R})$ satisfies the inequalities $(T_i)^T y + \mu = 0$ for all $i \in [n]$ iff $(y,\mu) = 0$. This assumption implies that there can not exist a line\footnote{A polyhedron $P \subseteq \mathbb{R}^m$ contains a line if there exists a vector $x \in P$ and a non-zero vector $d \in \mathbb{R}^m$ such that the vector $(x + \lambda d) \in P$ for all $\lambda \in \mathbb{R}$} in $Q$. This is because, for any polyhedron $P = \{Ax \le b\}$ to contain a line, there must exists a point $x \in P$ and a non-zero vector $d$ such that  $Ax + \lambda A d \le b$ for all $\lambda \in \mathbb{R}$ which is possible only if $A d = 0$. Since $Q$ does not contain a line, it must have at least one vertex (Theorem 2.6 of \cite{bertsimas}).
\end{proof}

\subsection*{Proof of Theorem \ref{convexextension} }
We will  show that for any $\alpha \ge 1$,  it can be efficiently determined whether there exists a convex function $f$ such that $f_i \le f(T_i) \le \alpha f_i$ for all $i \in [n]$.  Then by binary search  Approximate Extension can be solved efficiently.  By setting  $c_i = f_i$ and $d_i = \alpha f_i$ in the following Theorem, we can determine the same.
	  \begin{lemma}
	  	\label{gen-ext-convex}
	  	 Given a set $\mathcal{D} = \{T_1,\dots,T_n\}$ and  pairs $c_i,d_i$ ($c_i \le d_i$) associated with each $i \in [n]$,  it can be efficiently determined if there exists a convex function $f: \mathbb{R}^m \rightarrow \mathbb{R}$ such that $c_i \le f(T_i) \le d_i \quad \forall i \in [n]$.
	  \end{lemma}  
	  \begin{proof}
	   Let $\hat{g}$ be the convex roof function corresponding to the partial function $\{(T_1,d_1),\dots,(T_n,d_n)\}$. Our claim is that there exists a  convex function $f: \mathbb{R}^m \rightarrow \mathbb{R}$ such that $c_i \le f(T_i) \le d_i \quad \forall i \in [n]$ iff $c_i \le \hat{g}(T_i) \le d_i \quad \forall i \in [n]$.  Suppose $\hat{g}(T_i) \in [c_i,d_i]$ for all $i \in [n]$. Consider the function $\tf$ defined by (\ref{def:tf}). We know that $\tf$ is convex and  $\tf(T_i) = \hg(T_i) \in [c_i,d_i]$ for all $i \in [n]$. Thus $\tf$  is the required convex function. For the other direction, suppose for a contradiction that  
  $\hat{g}(T_i) \not \in [c_i,d_i]$ for some $i = i'$ and there exists a required convex function $f$. Therefore, we have $\hat{g}(T_{i'}) < c_{i'}$ as $\hat{g}(T_i) \le d_i$ for all $i \in [n]$ (since the partial function is $\{(T_1,d_1),\dots,(T_n,d_n)\}$).  Now consider a partial function $H' = \{(T_1,f(T_1)),\dots,(T_n,f(T_n))\}$ and its convex roof function $\hat{g'}$. Since $H'$ is clearly extendible,  we have $\hat{g'}(T_i) = f(T_i)$ for all $i \in [n]$ (Lemma \ref{convexext-iff}). Further from the definition of $\hat{g}$ and since $f(T_i) \le d_i$ for all $i$, it follows that $\hat{g'}(T_i) \le \hat{g}(T_i)$ for all $ i \in [n]$. Therefore we have $\hat{g'}(T_{i'}) \le \hat{g}(T_{i'}) < c_{i'} \le f(T_{i'})$.  This is a contradiction since as noted, $\hat{g'}(T_i) = f(T_i)$ for all $i \in [n]$. 
  
  This proves the theorem as the conditions  $c_i \le \hat{g}(T_i) \le d_i $ for all $i \in [n]$ can be checked efficiently.
	  \end{proof}

\subsection*{Proof of Lemma \ref{minimal-extension}}
 \iffalse
 Recall the definition of $\tf$ and $\tg$. Given a partial function   $H$, the function $\tf$ and $\tg$ are defined as:
\begin{align*}
\tf(x) = \max_{1 \le i \le N}   \{\displaystyle\sum_{j = 1}^{m} x_j y^i_j  + \mu^i \} \,  ,
\end{align*}
where  $N$ is the total number of vertices of $Q$ and $(y^i,\mu^i)$ is the $i$th vertex and  
\begin{align*}
\tg(x) =  \sup \left\{ \lambda \hg(y) + (1-\lambda) \hg(z) : x = \lambda y + (1-\lambda) z, y,z \in \conv(\mathcal{D}), \lambda \ge 1 \right\} \, .
\end{align*}
\fi
Fix $x$ outside  and $w$ inside  $\conv(\mathcal{D})$. Let $z \in \conv(\mathcal{D})$ and $ \lambda \ge 1$ be such that $x = \lambda w + (1-\lambda) z$. In the proof we vary $z$, and assume that $\lambda$ changes accordingly so that $x = \lambda w + (1-\lambda) z$.  Since $\hg$ is convex, the value of  $\lambda \hg(w) + (1-\lambda) \hg(z)$ increases as $z$ gets close to $w$. Therefore,  $\tg(x) = \sup_{w \in \conv(\mathcal{D})} L(w,x)$ where
\begin{align}
 L(w,x) = \lim_{z\to w} (\lambda \hg(w) + (1-\lambda)\hg(z)).
 \end{align}
  We know that $\hg(s)  = \max_{1 \le i \le N}   \{\sum_{j = 1}^{m} s_j y^i_j  + \mu^i\}$ (where $(y^i,\mu^i)$ is the vertex of the polyhedron $Q$ for all $i \in [N]$, as defined in main section). Let $1 \le k \le N$ be the index of the vertex that maximizes $\hg(w)$.
Therefore,  $\hg(s)  =   \{\sum_{j = 1}^{m} s_j y^k_j  + \mu^k\}$ at $s = w$ and in vicinity of  $w$ in line segment joining $z$ and $w$. Therefore,
\begin{multline*}
 L(w,x) = \lim_{z\to w}\lambda (\hg(w)-\hg(z)) + \hg(z) = \hg(w) + \lim_{z\to w} \lambda \sum_{j = 1}^{m} (w_j-z_j) y^k_j. 
 \end{multline*}
  Since $x - z = \lambda (w-z)$ so  we have $\lambda =   \frac{x_j-z_j}{w_j-z_j}$ for all $j \in [m]$. Now 
  \begin{multline*}
  L(w,x) = \hg(w) + \lim_{z\to w}\lambda \sum_{j = 1}^{m} (w_j-z_j) y^k_j =  \hg(w) + \lim_{z\to w} \sum_{j = 1}^{m} (x_j-z_j) y^k_j = \\ \hg(w) + \sum_{j = 1}^{m} (x_j-w_j) y^k_j =
  \sum_{j = 1}^{m} x_j y^k_j + \mu^k.
  \end{multline*}
   Thus for any $w \in \conv(\mcD)$, $L(w,x)$ is equal to $\sum_{j = 1}^{m} x_j y^k_j + \mu^k$ where  $1 \le k \le N$ is the index of the vertex that maximizes $\hg(w) = \{\sum_{j = 1}^{m} w_j y^k_j  + \mu^k\}$. This implies that $\tg(x) = \sup_{w \in \conv(\mathcal{D})} L(w,x)$ is at most $\tf(x) = \max_{1 \le i \le N}   \{\sum_{j = 1}^{m} x_j y^i_j  + \mu^i\}$.

Now we will show $\tg(x) \ge \tf(x)$. Let $1 \le t \le N$ be the index of the vertex that maximizes $\tf(x)$. Consider the vertex $(y^t,\mu^t)$. By definition of a vertex, there exists $c_1,\dots,c_m,c$ such that $\sum_{j=1}^{m} c_j y^t_j + c \mu^t > \sum_{j=1}^{m} c_j y_j + c \mu $ for all $(y,\mu) \in Q$. We claim that $c > 0$. First note that the left side of the above inequality is a finite quantity. Also  for $\mu = -\infty$ (sufficiently small),  the point $(y,\mu)$ is in $Q$ for any $y \in \mathbb{R}^m$. If $c < 0$ then the right side can be made  arbritray large by setting $\mu = -\infty$.
 For the case $c = 0$,  at least one of the $c_j$'s must be non zero. Therefore, by setting  $\mu = -\infty$  and $y_j < 0$ if $c_j < 0$ (and $y_j > 0$ if $c_j > 0$) the quantity $\sum_{j=1}^{m} c_j y_j$ can be made arbitrary large.
 %it is easy to see that if we set $\mu = -\infty$  then the point $(y,\mu)$ will be in $Q$ for any $y \in \mathbb{R}^m$. 
  %and $y_j < 0$ if $c_j < 0$ (and $y_j > 0$ if $c_j > 0$). Note that at least one of the $c_j$'s must be non zero.
   %the quantity $\sum_{j=1}^{m} c_j y_j$ can be made arbitrary large (note that at least one of the $c_j$'s must be non zero). 
   Therefore, we assume $c > 0$. 
   Now the above inequality can be written as  $\sum_{j=1}^{m} \frac{c_j}{c} y^t_j +  \mu^t > \sum_{j=1}^{m} \frac{c_j}{c} y_j +  \mu $ for all $(y,\mu) \in Q$. Let  $y = (\frac{c_1}{c},\dots,\frac{c_m}{c})$. The vector $y$ must be in $\conv(\mathcal{D})$  otherwise Convex-D should be unbounded (and hence right side should be unbounded), which is a contradiction. Now $\tg(x) \ge L(y,x)$ as $y \in \conv(\mathcal{D})$. Since  $\sum_{j=1}^{m} \frac{c_j}{c} y^t_j +  \mu^t > \sum_{j=1}^{m} \frac{c_j}{c} y_j +  \mu $ for all $(y,\mu) \in Q$, the function $\hg(s)  = \max_{1 \le i \le N}   \{\sum_{j = 1}^{m} s_j y^i_j  + \mu^i\}$ is given by $\hg(s)  =   \{\sum_{j = 1}^{m} s_j y^t_j  + \mu^t\}$ in a sufficiently small neighbourhood of $y$. Therefore, as before  $L(y,x) = \sum_{j = 1}^{m} x_j y^t_j + \mu^t$ and hence $\tg(x) \ge \tf(x)$.

\end{document}